\documentclass{elsarticle}
\IfFileExists{microtype.sty}{\usepackage{microtype}}{\relax}
\usepackage[british]{babel}
\usepackage{geometry}
\usepackage{amsthm}
\usepackage[cmex10]{amsmath}
\usepackage{amsfonts}
\usepackage{algpseudocode}
\usepackage{array}
\usepackage{enumerate}
\newcommand{\FF}{\mathbb F}

\newcommand{\NN}{\mathbb N}
\newcommand{\KK}{\mathbb K}
\newcommand{\iI}{\mathbb I}
\newcommand{\cG}{\mathcal G}
\newcommand{\cC}{\mathcal C}
\newcommand{\cW}{\mathcal W}
\newcommand{\cQ}{\mathcal Q}
\newcommand{\cP}{\mathcal P}

\newcommand{\fR}{\mathcal M}
\newcommand{\fmm}{\mathfrak m}

\newcommand{\fb}{\mathfrak b}
\newcommand{\fs}{\mathfrak s}
\newcommand{\fq}{\mathfrak q}
\newcommand{\PG}{\mathrm{PG}}
\newcommand{\hA}{\widehat{A}}
\newcommand{\hB}{\widehat{B}}
\newcommand{\ore }{\preceq}
\newcommand{\ors }{\prec}
\theoremstyle{theorem}
\newtheorem{theorem}{Theorem}[section]

\theoremstyle{definition}
\newtheorem{definition}{Definition}[section]

\begin{document}
\begin{frontmatter}
\title{Enumerative Coding for Line Polar Grassmannians with applications to codes}
\author[IC]{Ilaria Cardinali\corref{cor1}}
\ead{ilaria.cardinali@unisi.it}
\address[IC]{Department of Information Engineering and Mathematics, University of Siena,
Via Roma 56, I-53100, Siena, Italy}
\author[LG]{Luca Giuzzi}
\ead{luca.giuzzi@unibs.it}
\address[LG]{DICATAM - Section of Mathematics,
University of Brescia,
Via Branze 53, I-25123, Brescia, Italy}
\cortext[cor1]{Corresponding author.}



\begin{abstract}
A $k$-polar Grassmannian is a geometry having as pointset the set of all $k$-dimensional subspaces of a vector space $V$ which are totally isotropic for a given non-degenerate bilinear form  $\mu$ defined on $V.$
Hence it can be regarded as a subgeometry of the ordinary $k$-Grassmannian.
In this paper we deal with orthogonal line Grassmannians and with symplectic line Grassmannians, i.e. we assume $k=2$ and $\mu$ to be
 a non-degenerate symmetric or alternating form.
We will provide a method to efficiently enumerate the pointsets of  both orthogonal and symplectic line Grassmannians.
This has several nice applications; among them,
we shall discuss
an efficient encoding/decoding/error correction strategy for line
polar Grassmann codes of either type.
\end{abstract}

\end{frontmatter}

\section{Introduction}\label{s1}
Let $V$ be a vector space of dimension $n$ over a  field $\KK$.
For any $k<n$, the $k$--Grassmannian $\cG_{n,k}$
of $V$ is the geometry whose pointset $\cP(\cG_{n,k})$ consists of the $k$-dimensional subspaces of $V$
and whose lines are the sets of the form
\[ \ell_{X,Y}:=\{Z\colon  X<Z<Y\colon \dim Z=k\} \]
where $X$ and $Y$ are two subspaces of $V$ with $\dim(X)=k-1$ and $\dim(Y)=k+1.$ Incidence is containment.
It is well known that $\cG_{n,k}$ can be embedded,
as an algebraic variety $\mathbb{G}_{n,k}$,
into the projective space
$\PG(\bigwedge^kV)$, by means of the
Pl\"{u}cker embedding $e_k$. More precisely, $e_k$ maps the $k$-vector subspace $\langle v_1,v_2,\ldots,v_k\rangle$ of $V$ to the point $\langle v_1\wedge v_2\wedge\cdots\wedge v_k\rangle$ of $\PG(\bigwedge^k V)$  (see~\cite[Chapter VII]{HP1} and
\cite[Chapter XVI]{HP2} for details).

In this paper we assume $\KK$ to be a finite field $\FF_q$ of order $q.$ A basic problem  is to construct an
enumerator for the points of $\cG_{n,k}$, that is a bijection $\iota:\cP(\cG_{n,k})\to\{0,1,\ldots,N-1\}$, where $N:={n\brack k}_q$ is the number of points of $\cG_{n,k}.$
This has been studied extensively, see \cite{SE}, because of its
relevance to applications.
In particular, Grassmannians defined over finite fields have been used in coding theory
to construct linear projective codes~\cite{R1,R2} and network codes~\cite{KK}.
In general, it is not convenient to implement
a Grassmann code by na\"ively providing a generator matrix, as the number of columns is  large; so, a point enumerator $\iota$ provides an efficient way to uniquely identify subspaces corresponding to given positions. Several algorithms for enumerating the points of Grassmannians have been
proposed; see~\cite{M12,SE}.
We point out that, apart from their usefulness for coding theory, these algorithms also have independent geometric interest.

The present paper is concerned with the problem outlined above, focusing  on the case of line polar Grassmannians of either orthogonal or symplectic type.
These are proper subgeometries of line Grassmannians having as pointset
the set of lines of a vector space $V$ which are totally isotropic for
a given non-degenerate symmetric or alternating form.
We shall determine a method to enumerate these lines, following the basic approach of~\cite{Cover}.
In general, the case of polar Grassmannians is more involved than that of ordinary $k$-Grassmannians,
as there are some requirements imposed on the subspaces which have to be fulfilled; thus, a careful use of linear algebra (combined with combinatorial techniques) is necessary in order to get a reasonably efficient
representation.  We shall also examine in detail the computational complexity of the proposed algorithms.

Our algorithms will then be applied to polar Grassmann codes of either orthogonal or symplectic type.
These codes have been introduced respectively in~\cite{IL13} and in~\cite{IL15} as
linear codes arising from the Pl\"ucker embedding of polar Grassmannians of
orthogonal or symplectic type.
Some bounds on their minimum distance have been obtained:
it has been proved in~\cite{ILP14} for $q$ odd and in~\cite{IL16} for $q$
even that
if $k=2$ then
the minimum distance of a line orthogonal Grassman code is $q^{4n-5}-q^{3n-4}$;
the minimum distance of a line symplectic Grassman code is  $q^{4n-5}-q^{2n-3}$ (see~\cite{IL15}).

We shall set the notation in Section~\ref{Notation}
and recall some notions about polar Grassmann codes in
 Section~\ref{PGC}.
The organization of the paper and the main results are outlined in Section~\ref{organization}.

\subsection{Notation}\label{Notation}
Let $V:=V(2n+1,q)$ be a vector space of dimension
$2n+1$ over a finite field $\FF_q$ and let $\fq:V\to\FF_q$ be a non-degenerate
quadratic form. 
The $k$-orthogonal Grassmannian $\Delta_{n,k}$ is a point-line
geometry whose
points are all the totally $\fq$-singular $k$-subspaces of $V$ and
whose lines are the sets either
of the form
\[ \ell_{X,Z}:=\{ Y\in\cG_{2n+1,k}: X<Y<Z \}\text{\,\,for\,\,} k<n, \]
where $\dim X=k-1$, $\dim Z=k+1$ and $Z$ is totally $\fq$-singular or
of the form
 \[ \ell_{X}:=\{ Y:  X<Y<X^{\perp\fq}, \dim Y=n, \text{$Y$ totally singular } \}\text{\,\,for\,\,} k=n,\]
with $\dim X=n-1$ and
$X^{\perp\fq}$
the space orthogonal to $X$ with respect to the bilinear form
$\fb$ associated with $\fq$.
Incidence is defined in the natural way.

Likewise, denote by $\overline{V}:=V(2n,q)$ a vector space of dimension
$2n$ over a finite field $\FF_q$ and consider a non-degenerate
 alternating bilinear form  $\fs:\overline{V}\times \overline{V}\to\FF_q$.
The $k$-symplectic Grassmannian $\overline{\Delta}_{n,k}$ has
as points all totally $\fs$-isotropic $k$-spaces of $\overline{V}$ and as lines
the sets of the form
\[ \ell_{X,Z}=\{ Y\in\cG_{2n,k}: X<Y<Z \} \text{\,\,for\,\,} k<n, \]
with $\dim X=k-1$, $\dim Z=k+1$ and $Z$ totally $\fs$-isotropic
or
\[ \ell_{X}=\{ Y: X<Y, \dim Y=n, Y \text{ totally isotropic} \} \text{\,\,for\,\,} k=n,\]
with $\dim X=n-1$.

By construction,
any point of $\Delta_{n,k}$ is also a point of $\cG_{2n+1,k}$ and
any point of $\overline{\Delta}_{n,k}$ is also a point of $\cG_{2n,k}$.

Consider the embeddings $\varepsilon_k:=e_k|_{\Delta_{n,k}}$ and $\overline{\varepsilon}_k:=e_k|_{\overline{\Delta}_{n,k}}$ induced by the Pl\"{u}cker embedding $e_k$ on
the polar Grassmannians of orthogonal and symplectic type.
In particular,
$\varepsilon_k({\Delta_{n,k}})$ is a subvariety of $\mathbb{G}_{2n+1,k}$
and $\overline{\varepsilon}_k({\overline{\Delta}_{n,k}})$ is a subvariety of
$\mathbb{G}_{2n,k}$. We summarize what is known about  these embeddings.
We warn the reader that we shall always use vector dimensions.
\begin{theorem}[\cite{IP13}]
 Let $\varepsilon_k:\Delta_{n,k}\to\PG(\bigwedge^kV)$ be the
  restriction of the Pl\"ucker embedding to the orthogonal polar
  Grassmannian $\Delta_{n,k}$ and let $W_{n,k}:=\langle\varepsilon_k(\Delta_{n,k})\rangle$. Then,
     \begin{itemize}
     \item For $k<n$, $\varepsilon_k$ is projective and
           \[ \dim W_{n,k}=\begin{cases}
               {{2n+1}\choose{k}} & \text{ if $\mathrm{char}(\FF)$ odd } \\
               \binom{2n+1}{k}-\binom{2n+1}{k-2} & \text{ if $\mathrm{char}(\FF)=2$.}\\
           \end{cases}\]
   \item For $k=n$, $\varepsilon_n:\Delta_{n,n}\to \PG(W_{n,n})$
     maps lines into conics.
   \end{itemize}
 \end{theorem}
 \begin{theorem}[\cite{B,PS}]
  Let $\overline{\varepsilon}_k:\Delta_{n,k}\to\PG(\bigwedge^kV)$ be the
  restriction of the Pl\"ucker embedding to the symplectic polar
  Grassmannian $\overline{\Delta}_{n,k}$ and let $\overline{W}_{n,k}:=\langle\overline{\varepsilon}_k(\overline{\Delta}_{n,k})\rangle$. Then,
  \begin{itemize}
  \item
     $\overline{\varepsilon}_k:\overline{\Delta}_{n,k}\to\PG(\bigwedge^kV)$
     is projective and
     $\dim(\overline{W}_{n,k})={{2n}\choose{k}}-{{2n}\choose{k-2}}$.
  \end{itemize}
\end{theorem}

\subsection{Polar Grassmann Codes}
\label{PGC}
Throughout the paper, we shall denote
by $N$ the length of a linear code and by $K$ its dimension; as before,
lower case letters $n$ and $k$ shall be used to represent the parameters
of the associated Grassmannians.

A $q$-ary code $\cC$ of length $N$ and dimension $K$ is called \emph{projective} if the
columns of its generator matrix are the coordinates of $N$ distinct
points in $\PG(K-1,q)$.
Conversely, given a set of $N$  distinct points
$\Omega=\{P_1,\ldots,P_N\}$ in $\PG(W)$
we call \emph{projective code induced by $\Omega$} any linear
code
 $\cC(\Omega)$ generated by a matrix $G$ whose columns consist
of the coordinates of the points in $\Omega$ with respect to some reference system.
We have $K=\dim\langle\Omega\rangle$.
Clearly,
$\cC(\Omega)$ is defined only up to code equivalence, but in the rest
of this paper we shall speak, with a slight abuse of notation, of
\emph{the} code induced by $\Omega$; see~\cite{tvn}
for more details.

The close correspondence between hyperplane sections of $\Omega$ and the weights of the codewords of $\cC(\Omega)$ is a basic result of the theory of projective codes. In particular,
hyperplanes of $\PG(W)$ having maximal proper intersection with $\Omega$
are associated with codewords of minimum weight.

The projective codes $\cC_{n,k}$
arising from  the pointset
$e_k(\cG_{n,k})\subseteq\PG(\bigwedge^kV)$ are called
Grassmann codes. They have been introduced in
\cite{R1,R2} as a generalization of Reed-Muller codes of the first
order and have been widely investigated ever since: both their
monomial automorphism groups and minimum weights are well understood,
see~\cite{GK2013,GL2001,GPP2009,KP13,N96,R3}.
Codes associated with subsets of Grassmannians have also been studied; see,
for instance \cite{BGH}.
We point out that Grassmannians can also be used to obtain Tanner codes;
see \cite{Pt}.

All these codes have a fairly low rate; as such,  in order to be efficiently
implemented, it is paramount to  provide some encoding and decoding
algorithms acting locally on the components.
To this aim, in~\cite{SE} an enumerative coding scheme for
 Grassmannians is considered and some efficient algorithms are presented; see also~\cite{M12}.
\par
Starting with~\cite{IL13} we have been considering
linear codes arising from the Pl\"ucker embedding of polar Grassmannians
of either orthogonal or symplectic type.
In particular, in~\cite{IL13} a new family of linear codes
 related to the Pl\"{u}cker embedding of polar orthogonal Grassmannians
$\Delta_{n,k}$ has been introduced and some bounds on its minimum distance
have been determined.

In close analogy to orthogonal polar Grassmann codes,  in~\cite{IL15}  we introduced
symplectic polar Grassmann codes, that is codes  arising from
the Pl\"ucker embedding of a symplectic Grassmannian.

Either family of polar Grassmann codes can be
obtained from a Grassmann code $\cC_{2n+1,k}$ or $\cC_{2n,k}$
by just deleting all the columns corresponding respectively
to $k$--spaces which
are non-singular with respect to $\fq$  or non-isotropic with
respect to $\fs$ --- as such they can be regarded in a natural way as
punctured versions of $\cC_{2n+1,k}$ or $\cC_{2n,k}$.
We summarize in the following theorems what is currently known about
the parameters of these codes.
  \begin{theorem}[\cite{IL13},\cite{ILP14},\cite{IL16}]\label{pgc-thm}
    The known parameters $[N,K,d]$ of $\cP_{n,k}:=\cC(\Delta_{n,k})$
    are\vskip-0.1cm
    \[
      \begin{array}{c|c|c|c|c}
        (n,k) & N & K & d & \text{Reference} \\ \hline
        1\leq k<n&  \displaystyle\prod_{i=0}^{k-1}\frac{q^{2(n-i)}-1}{q^{i+1}-1}
                  & \displaystyle\binom{2n+1}{k} & d\geq \widetilde{d}(q,n,k) & \text{\cite{IL13}} \\
        \hline
        (3,3)  & (q^3+1)(q^2+1)(q+1) & 35 & q^2(q-1)(q^3-1) & \text{\cite{IL13}} \\ \hline
    (n,2) & \frac{(q^{2n}-1)(q^{2n-2}-1)}{(q-1)(q^2-1)}  &  (2n+1)n    & q^{4n-5}-q^{3n-4} & \text{\cite{ILP14}}  \\ \hline
      \end{array} \]

   \centerline{\small $q$ odd}
\[
 \begin{array}{c|c|c|c|c}
     (n,k) & N & K & d & \text{Reference} \\ \hline
  1\leq k<n&  \displaystyle\prod_{i=0}^{k-1}\frac{q^{2(n-i)}-1}{q^{i+1}-1}
            &  \displaystyle\binom{2n+1}{k}-\binom{2n+1}{k-2} &
             d\geq \widetilde{d}(q,n,k) & \text{\cite{IL13}}
                                                   \\ \hline
   (3,3)  & (q^3+1)(q^2+1)(q+1) & 28 & q^5(q-1) & \text{\cite{IL13}} \\ \hline
   (n,2) & \frac{(q^{2n}-1)(q^{2n-2}-1)}{(q-1)(q^2-1)}  &  (2n+1)n-1    & q^{4n-5}-q^{3n-4} & \text{\cite{IL16}}  \\ \hline

 \end{array}\] 
   \centerline{\small $q$ even}
{\small
   \[ \widetilde{d}(q,n,k):=(q+1)(q^{k(n-k)}-1)+1 \]
}
\end{theorem}

 \begin{theorem}[\cite{IL15}]
    The known parameters $[N,K,d]$ of $\cW_{n,k}:=\mathcal{C}(\overline{\Delta}_{n,k})$
    are
      \[
      \begin{array}{c|c|c|c|c}
        (n,k) & N & K & d & \text{Reference} \\
              \hline
1<k\leq n &  \prod_{i=0}^{k-1}(q^{2n-2i}-1)/(q^{i+1}-1) & {2n\choose k}-{2n\choose{k-2}} & & \text{\cite{IL15}} \\ \hline
 (n,2) &  \frac{(q^{2n}-1)(q^{2n-2}-1)}{(q-1)(q^2-1)} & n(2n-1)-1 & q^{4n-5}-q^{2n-3} & \text{\cite{IL15}} \\ \hline
(3,3)  &  (q^3+1)(q^2+1)(q+1) & 14  & q^6-q^4 & \text{\cite{IL15}} \\\hline
      \end{array} \]
\end{theorem}

\subsection{Organization of the paper and Main Results }\label{organization}
 In Section~\ref{p2} we recall the notion of prefix enumeration and describe counting algorithms
for the points of both $\Delta_{n,2}$ and $\overline{\Delta}_{n,2}$.
In particular, in \S~\ref{s2} we consider
the number of totally $\fq$-singular lines of $V$
spanned by  vectors with a prescribed prefix, while in \S~\ref{w2} we investigate
the totally $\fs$-isotropic lines of $\overline{V}$.
The complexity of the prefix enumerators is discussed in \S~\ref{s4}.
\begin{theorem}
\label{t15}
For an orthogonal line Grassmannian, the computational complexity for determining the number of
  points whose representation
  begins with a prescribed prefix is $O(n^2)$.

For a symplectic line Grassmannian, the computational complexity for determining the number of
  points whose representation
  begins with a prescribed prefix is $O(n)$.
\end{theorem}
These results are used in
Section~\ref{s3} to present
an enumerative coding scheme according to the approach of~\cite{Cover}.
In \S~\ref{compl} we
analyze the overall complexity of our enumerative encoding scheme.
\begin{theorem}\label{main thm}
The computational complexity of the point enumerator of an orthogonal line Grassmannian is $O(q^2n^3)$.
The computational complexity of the point enumerator of a symplectic line Grassmannian is $O(q^2n^2)$.
\end{theorem}
Section~\ref{s5} is dedicated to applications of the scheme
introduced in Section~\ref{s3} to orthogonal and symplectic line polar Grassmann codes.
We propose some encoding/decoding and
error correction strategies which act locally on the components of
the codewords.

\section{Prefix enumeration}
\label{p2}
In this section we shall present an algorithm  to count
the number of points of a line polar Grassmannian whose representation
satisfies certain conditions.
This will be
essential for the enumerative encoding algorithm of Section~\ref{s3}.

\subsection{Preliminaries}
\label{sec2.1}
In order to simplify the exposition, in this section we shall slightly alter the notation introduced in
Section~\ref{Notation}. For $\varepsilon=0,1$, let
$V^{\varepsilon}:=V(2n+\varepsilon,q)$ be a vector space of dimension $2n+\varepsilon$ over $\FF_q$  and let
$B_{\varepsilon}$ be a fixed basis of $V^{\varepsilon}$.
So, according to Section~\ref{Notation}, $V^0:=\overline{V}$ and
$V^1:=V$.
Up to projectivities,
there is exactly one class of non-degenerate
quadratic forms on $V^1$; hence,
it is not restrictive to
choose the following quadratic form $\fq$:
\begin{equation}\label{equation quadratic form}
 \fq(\mathbf{x})=x_1^2+\sum_{i=1}^{n}x_{2i}x_{2i+1},
\end{equation}
where $x=(x_i)_{i=1}^{2n+1}$.
The associated bilinear form $\fb$ is
  \[ \fb({x},{y}):=2x_1y_1+\sum_{i=1}^{n}\left(
    x_{2i}y_{2i+1}+y_{2i}x_{2i+1}\right), \]
where ${x}=(x_i)_{i=1}^{2n+1}$, ${y}=(y_i)_{i=1}^{2n+1}$.
Note that, for $q$ even, the form $\fb$ is alternating
and degenerate, while for $q$ odd $\fb$ is non-degenerate and
symmetric. We will denote by $\mathcal{Q}$ the non-degenerate parabolic quadric of $\PG(V)$ defined by $\fq.$

If $\varepsilon=0$, consider the following
non-degenerate
symplectic form $\fs:\overline{V}\times \overline{V}\to\FF_q$,
\begin{equation}\label{fseq} \fs({x},{y})=\sum_{i=1}^{n} (x_{2i-1}y_{2i}-y_{2i-1}x_{2i}) \end{equation}
where ${x}=(x_i)_{i=1}^{2n}$, ${y}=(y_i)_{i=1}^{2n}$.
We will denote by $\mathcal{W}$ the non-degenerate symplectic polar space of $\PG(\overline{V})$ defined by $\fs.$

Recall that a $(2\times t)$-matrix $G$ is
said to be in \emph{Hermite normal form} or in \emph{row reduced
echelon form} (RREF, in brief) if it is in row-echelon form, the leading
non-zero entry of each row is $1$ and all entries above a leading entry
are $0$.
For each line $\ell$  of $\PG(V^{\varepsilon})$,
there are two uniquely determined vectors $X,Y\in\FF_q^{2n+\varepsilon}$ such that $\ell=\langle X,Y\rangle$
and $G_{\ell}:=\begin{pmatrix} X\\Y\end{pmatrix}$ is a
$2\times (2n+\varepsilon)$-matrix in RREF. We call $G_{\ell}$ the \emph{representation} of $\ell$.

We remind that a line $\ell=\langle X,Y\rangle$ of $\PG(V^1)$ is said to be
\emph{totally $\fq$-singular} if $\fq(X)=\fq(Y)=0=\fb(X,Y)$.
Likewise, a line $\ell=\langle X,Y\rangle$ of $\PG(V^0)$ is
\emph{totally $\fs$-isotropic} if $\fs(X,Y)=0$.
\begin{table}
\caption{Useful numbers}
\[
\begin{array}{l|l|l}
\Xi & |\Xi|_1:=|\{\text{points of}\ \Xi \}| & |\Xi|_2:=|\{\text{lines of}\ \Xi \}| \\
\hline
\PG(v,q) & \frac{(q^{v+1}-1)}{q-1}& \frac{(q^{v}-1)(q^{v+1}-1)}{(q^2-1)(q-1)}\\
\hline
Q(2v,q) & \frac{(q^{2v}-1)}{q-1} &  \frac{(q^{2v-1}-1)(q^{2v}-1)}{(q^2-1)(q-1)} \\
\hline
Q^+(2v-1,q) & \frac{(q^{v}-1)(q^{v-1}+1)}{q-1} & \frac{(q^{2v-2}-1)(q^{v}-1)(q^{v-1}+1)}{(q^2-1)(q-1)}\\
\hline
W(2v-1,q) & \frac{q^{2v}-1}{q-1} & \frac{(q^{2v}-1)(q^{2v-2}-1)}{(q-1)(q^2-1)} \\
\end{array}\]
\label{unum}
\end{table}

 Denote by ${\fR}_{2,t}$ the set of all
$(2\times t)$-matrices over $\FF_q$ and also let
\[ \fR_{2}^{\varepsilon}:=\bigcup_{t=0}^{2n+\varepsilon}\fR_{2,t}. \]
with $\varepsilon\in\{0,1\}$ and $\fR_{2,0}:=\{\emptyset\}$.
\noindent\par
For $1\leq t\leq 2n+\varepsilon$, let $\displaystyle S_t=\begin{pmatrix}A_t
\\
B_t\end{pmatrix}\in \fR_2^{\varepsilon}$
with $A_t:=(\alpha_1,\alpha_2,\ldots,\alpha_t)$ and $B_t:=(\beta_1,\beta_2,\ldots,\beta_t)$
and
put $\hA:=(A_t,x_{t+1},x_{t+2},\ldots,x_{2n+\varepsilon})$ and $\hB:=(B_t,y_{t+1},y_{t+2},\ldots,y_{2n+\varepsilon}).$

Then we say that $S_t$ is the \emph{t-prefix} or the
\emph{t-leading part} of the $2\times(2n+\varepsilon)$-matrix $\begin{pmatrix} \hA\\ \hB\end{pmatrix}$.
The \emph{length} of $S_t$ is the number $t$ of its columns.
\par\noindent
We shall also define the following two vectors of $V^{\varepsilon}$: $A:=(A_t,0,\ldots,0)$ and
$B:=(B_t,0,\ldots,0)$.
\begin{definition}
\label{nqns}
Let
\begin{itemize}
\item
  $n_{\fq}:\fR_2^1\times\NN\to\NN$ be the function sending any $(S_t,n)\in \fR_2^1\times\NN$ to  the number of totally $\fq$-singular lines of $\PG(2n,q)$
 whose representation in RREF has prefix $S_t$;
\item
   $n_{\fs}:\fR_2^0\times\NN\to\NN$ be the function sending any $(S_t,n)\in \fR_2^0\times\NN$ to the number of totally $\fs$-isotropic lines of $\PG(2n-1,q)$
 whose representation in RREF has prefix $S_t$.
\end{itemize}
If $S_t$ is not in RREF, we have $n_{\fq}(S_t,n)=0$ and $n_{\fs}(S_t,n)=0$ for any $n\in \NN.$
Henceforth, we shall always silently assume that $S_t$ is given in RREF.
\end{definition}

\begin{definition}\label{t-REF}
We say that a $(2\times r)$-matrix
\[ S_r=\begin{pmatrix}
  \alpha_1 & \ldots & \alpha_t & \ldots & \alpha_r \\
  \beta_1 & \ldots & \beta_t   & \ldots & \beta_r
  \end{pmatrix}\]
 is in
\emph{$t$-Row Echelon Form} (in brief $t$-REF) if one of the following two
conditions holds
\begin{enumerate}[a)]
  \item $\alpha_t=\beta_t=0$ and $S_r$ is in RREF, or
  \item $\alpha_t=0$ or $\beta_t=0$ but $(\alpha_t,\beta_t)\neq (0,0)$,
    $S_r$ is in row-echelon form and the leading non-zero entry in each row is $1$.
\end{enumerate}
\end{definition}
Note that, in general, a matrix in $t$-REF is not in RREF.
Indeed, given a $(2\times r)$-matrix $S$ in RREF,
if either $\alpha_t=0$ or $\beta_t=0$, then
$S$ is already also in $t$-REF;
otherwise, when $\beta_t\neq0$,
we can always
subtract from the first row of $S$ the second row multiplied by
$\lambda=\alpha_t\beta_t^{-1}(\neq 0)$  to get a new matrix
\begin{equation}
\label{SpC}
 S'=\begin{pmatrix}
  \alpha_1-\lambda\beta_1 & \ldots & 0 & \ldots & \alpha_{r}-\lambda\beta_r \\
    \beta_1 & \ldots &  \beta_t & \ldots & \beta_r
    \end{pmatrix} \text{\rm in\,\, $t$-REF.}
\end{equation}
It is now easy to see that, for any line $\ell$ and any $1\leq t\leq 2n+\varepsilon$, there exists
exactly one matrix in $t$-REF whose rows span $\ell$.

For $q=2^s$ denote by $\mathrm{Tr}_2(x)$ the absolute trace of $x\in\FF_q$,
 that is $\mathrm{Tr}_2(x):=\sum_{i=0}^{s-1} x^{2^i}.$

\subsection{Enumerating orthogonal Grassmannians}
\label{s2}
In this section we shall compute the enumerating function
$n_{\fq}$ introduced in Definition~\ref{nqns}.
If $S_t$ ($1\leq t\leq 2n+1$) is not in RREF, then $n_{\fq}(S_t,n)=0$ for all $n$.
Suppose now $(S_t,n)\in\fR_{2}^{1}\times\NN$ with $S_t\in\fR_{2,t}^{1}$
in RREF.
The value $n_{\fq}(S_t,n)$ is
the number of solutions in the unknowns $x_i$ and $y_i$,
$i=t+1,\dots, 2n+1$, of the system of quadratic equations
\begin{equation}
\label{mains}
\begin{cases}
  \fq(\hA)=0 \\
  \fq(\hB)=0 \\
  \fb(\hA,\hB)=0.
\end{cases}
\end{equation}
The first step of the algorithm is to transform $S_t$ in
$t$-REF, see Definition~\ref{t-REF}.

We will distinguish two
cases, depending on the parity of $t$.
These cases will not be fully
independent:
as it will be seen, our algorithm for $t$ even requires some computations
with some auxiliary prefixes of odd length and, likewise,
some cases with a prefix of odd length are dealt with
by reducing to different cases where the prefix has an even number of columns.
In any case, as the analysis shall show, this
will not lead to an infinite recursion and will ultimately provide the
correct value without explicitly solving~\eqref{mains}.

\subsubsection{Even $t$}\label{even-t}
If $t=0$, then
$n_{\fq}(\emptyset,n)=\frac{(q^{2n-1}-1)(q^{2n}-1)}{(q^2-1)(q-1)}$
 is the number of the (totally singular) lines of
$\mathcal{Q}$;  see Table~\ref{unum}.
For $t>0$ even, System~\eqref{mains} can be explicitly written as follows.
\begin{equation}
\setlength{\nulldelimiterspace}{0pt}
\label{e3-even}
\begin{cases}
{\displaystyle
\alpha_1^2 + \sum_{i=1}^{t/2-1}\alpha_{2i}\alpha_{2i+1}+\alpha_tx_{t+1}+\!\!\!\sum_{i=t/2+1}^{n}\!\!x_{2i}x_{2i+1}=0}\\
{\displaystyle \beta_1^2 + \sum_{i=1}^{t/2-1}\beta_{2i}\beta_{2i+1}+\beta_ty_{t+1}+\!\!\!\sum_{i=t/2+1}^{n}\!\!y_{2i}y_{2i+1}=0}\\
2\alpha_1\beta_1+ \alpha_ty_{t+1}+\beta_tx_{t+1}+\!{\displaystyle\sum_{i=1}^{t/2-1}}(\alpha_{2i}\beta_{2i+1} +\alpha_{2i+1}\beta_{2i})+
\smash{\!\!\!{\displaystyle\sum_{i=t/2+1}^{n}}\!\!(x_{2i}y_{2i+1} +x_{2i+1}y_{2i})=0.}
  \end{cases}
\end{equation}
We will compute the number of solutions of the system~(\ref{e3-even}) in the unknowns $x_i$, $y_i$, for $t+1\leq i,j\leq 2n+1$.  We distinguish several cases.
\begin{enumerate}[\mbox{A.}1)]
\item\label{A1q}
\fbox{$\alpha_t=0$ and $\beta_t\neq 0$.}
The second and third equations of~\eqref{e3-even} are linear in respectively
$y_{t+1}$ and $x_{t+1}$ and
the coefficient of $y_{t+1}$ is non-zero.
 So, for any choice of $(y_{t+2},\ldots,y_{2n+1})\in\FF_q^{2n-t}$, the value of
$y_{t+1}$ is uniquely determined; there are $q^{2n-t}$ possibilities. Similarly, the third equation directly
provides the value of $x_{t+1}$ once $x_{t+2},\ldots,x_{2n+1},$ satisfying the
first equation, are given. Hence, there remains to study the number of solutions of the first equation,
which can be written as
\begin{equation}
\label{e3p1p}
\fq(A)+\!\!\!\sum_{i=t/2+1}^{n}\!\!x_{2i}x_{2i+1}=0.
\end{equation}
As $S_t$ is in $t$-REF, we have $A_t=(\alpha_1,\ldots,\alpha_{t-1},0)\neq\mathbf{0}$;
i.e. there is $i<t$ such that $\alpha_i\neq 0$.
Since $B_t=(\beta_1,\ldots,\beta_t)\neq\mathbf{0}$, any vector solution of~\eqref{e3p1p}, with
arbitrary choices of $y_{t+2},\ldots,y_{2n+1}$ gives different lines.
Call $\widehat{\eta_0}(A)$ the number of solutions of~\eqref{e3p1p}.
If $\fq(A)=0$, then $\widehat{\eta_0}(A)$ is the number
of vectors $(x_{t+2},\ldots,x_{2n+1})\in \FF_q^{2n-t}$ satisfying $\fq^+(x_{t+2},\ldots,x_{2n+1})=0$ where
\[ \fq^+(x_{t+2},\ldots,x_{2n+1}):=\sum_{i=t/2+1}^n x_{2i}x_{2i+1}. \]
Hence $\widehat{\eta_0}(A)$ is $(q-1)$ times the number of points of a hyperbolic quadric $\cQ^+$ in $\PG(2n-t-1,q)$. If $\fq(A)\neq0$, then $\widehat{\eta_0}(A)$ is the number of
vectors $(x_{t+2},\ldots,x_{2n+1})\in \FF_q^{2n-t}$ such that
\[ \fq^+(x_{t+2},\ldots,x_{2n+1})=-\fq(A). \]
If $q$ is odd, then  the form $\fq^+$ has
the same quadratic character as $-\fq(A)$
for half of the points of $\PG(2n-t-1,q)$ not in  $\cQ^+$;
each of these points
contributes $2$ vector solutions of~\eqref{e3p1p}. The points with
quadratic character different from that of $-\fq(A)$ do not contribute
any solution.
Thus, $\widehat{\eta_0}(A):=\eta_0(\fq(A))$, where
\begin{equation}\label{eta_0}
 \eta_0(c):=\begin{cases}
  (q-1)|Q^+(2n-t-1,q)|_1+1 & \text{if $c=0$}\\
2\cdot\frac{1}{2}\left(|\PG(2n-t-1,q)|_1- |Q^+(2n-t-1,q)|_1\right) & \text{if $c\neq 0$}, \\
  \end{cases}
\end{equation}
 that is, by Table~\ref{unum}:
 \[\eta_0(c)=\begin{cases}
(q-1)\cdot\frac{(q^{n-t/2}-1)(q^{n-t/2-1}+1)}{q-1}+1 & \text{if $c=0$} \\[6pt]
  \frac{q^{2n-t}-1}{q-1}-\frac{(q^{n-t/2}-1)(q^{n-t/2-1}+1)}{q-1}
  &
\text{if $c\neq 0$}.\\
\end{cases} \]
For $q$ even, an analogous argument, where we consider the absolute trace $\mathrm{Tr}_2(\fq(A))$ of $\fq(A)$ instead of
its quadratic character, leads to the same formula~\eqref{eta_0}.

Finally,
\[ n_{\fq}(S_t,n)=\underbrace{|\{\text{solutions to \eqref{e3p1p}}\}|}_{\text{possibilities for } x_{t+2},\ldots,x_{2n+1}}
\times
{\underbrace{q^{2n-t}}_{{\begin{minipage}{1.7cm}\tiny\begin{centering} possibilities for \\ $y_{t+2},\ldots,y_{2n+1}$\end{centering}\end{minipage}}}}\!\!\!\!\!=\eta_0(\fq(A))\cdot q^{2n-t}.
 \]

\item\label{A2q} \fbox{$\alpha_t\neq 0$ and $\beta_t=0$.} This case is analogous
to A.\ref{A1q} with the roles of the first and the second equation
reversed. The only difference is for $B=\mathbf{0}$.
Indeed,
\begin{enumerate}[\mbox{A.\ref{A2q}}.1)]
\renewcommand{\theenumi}{\relax}
 \item\label{A21} for $B\neq\mathbf{0}$ and $\beta_t=0$,
   we argue exactly as in~A.\ref{A1q} and
   $n_{\fq}(S_t,n)=\eta_0(\fq(B))\cdot q^{2n-t}$.
 \item\label{A22} for $B=\mathbf{0}$ we first count the number of points
   of the hyperbolic quadric having equation
   $\fq^+(y_{t+2},\ldots,y_{2n+1}):=y_{t+2}y_{t+3}+\ldots+y_{2n}y_{2n+1}=0$.
Let $\ell=\langle\hA,\hB\rangle$ be any line with $\hB$ given by the previous equation and denote by $i>t$ the index of the first non-zero component $y_i$ of $\hB$. Then,  $\begin{pmatrix}\hA-x_iy_i^{-1}\hB \\ \hB\end{pmatrix}$ is the representative matrix of $\ell$ in RREF. In particular, $x_i=0$ and the $t$-prefix of this matrix is the same as that of $\begin{pmatrix}\hA\\ \hB\end{pmatrix}$.
   So, there are $q^{2n-t-1}$ possibilities
   for $x_{t+1},\ldots,x_{2n+1}$. Thus,
   \[ n_{\fq}(S_t,n):=\frac{q^{2n-t-1}}{q-1}(q^{n-t/2}-1)(q^{n-t/2-1}+1). \]
\end{enumerate}
\item \fbox{$\alpha_t=\beta_t=0$.}\label{A3q}
In this case the matrix $G=\begin{pmatrix}\hA\\ \hB\end{pmatrix}$ has the form
\[ G=\begin{pmatrix}
  \alpha_1 & \ldots & \alpha_{t-1} & 0 & x_{t+1} & \ldots x_{2n+1} \\
  \beta_1 & \ldots & \beta_{t-1} & 0 & y_{t+1} & \ldots y_{2n+1} \\
\end{pmatrix}.
\]
As the coefficients of  $x_{t+1}$ and $y_{t+1}$ in the equations of~(\ref{e3-even}) are both zero, the
system~\eqref{e3-even} is formally
the same as the system defined by
\[ G'=\begin{pmatrix}
  \alpha_1 & \ldots & \alpha_{t-1} &  x_{t+2} & \ldots x_{2n+1} \\
  \beta_1 & \ldots & \beta_{t-1} &  y_{t+2} & \ldots y_{2n+1} \\
\end{pmatrix}.
\]
We shall call ``reduced'' this new system, where the unknowns $x_{t+1}$ and $y_{t+1}$ have
been removed.
It is straightforward to see that for each solution of the reduced system
 there are $q^2$ solutions of~\eqref{e3-even}, being $x_{t+1}$ and $y_{t+1}$ arbitrary.
  Note that the number of solutions
 of the reduced system is $n_{\fq}(S_{t-1},n-1)$ where
 \[ S_{t-1}:=\begin{pmatrix}
    \alpha_1 & \alpha_2 & \ldots & \alpha_{t-1} \\
    \beta_1 & \beta_2 & \ldots & \beta_{t-1}
\end{pmatrix}. \]

We now consider  three subcases:
\begin{enumerate}[\mbox{A.\ref{A3q}.}1)]
\renewcommand{\theenumi}{\relax}
\item\label{A3p1}
\fbox{$A_t=B_t=\mathbf{0}$.} In this case,
 $n_{\fq}(S_t,n)$ is the number of lines contained in the parabolic quadric $\cQ$ defined
 by $\fq$ and in
 the subspace $\Pi$ of codimension $t$ described by the equations
 \[ x_1=0, x_2=0, \ldots, x_{t}=0. \]
As
 $\cQ':=\Pi\cap\cQ$ is a cone of vertex
$W=(\overbrace{0,0,\ldots,0}^{t},1,0,\ldots,0)$ and
basis the hyperbolic quadric $\cQ^+$ of $\PG(2n-t-1,q)$ with equation
 \[ \begin{cases}
   x_1=x_2=\ldots=x_t=0 \\
   x_{t+2}x_{t+3}+x_{t+4}x_{t+5}+\cdots+x_{2n}x_{2n+1}=0,
   \end{cases}\]
we have $n_{\fq}(S_t,n)=\sigma q^2+|\cQ^+|_1$, where $\sigma= |\cQ^+|_2$ is the number of lines of $\cQ^+$;
see Table~\ref{unum}.
\item\label{A3p2} \fbox{$A_t\neq\mathbf{0}$ and
  $B_t=\mathbf{0}$.} In this case,
  $n_{\fq}(S_t,n):=
  q^2n_{\fq}(S_{t-1},n-1)+\sigma_1$, where
  $\sigma_1$ corresponds to the number of solutions of~\eqref{e3-even} which do not
  arise from solutions of the reduced system. This happens only if the second row of
  $G'$ is null, but the second row of $G$ is not.
  That is,
  \[ G=\begin{pmatrix}
     A_t&  0   & {x}_{t+2} & \ldots & {x}_{2n+1} \\
    \mathbf{0} &   1   & 0 & \ldots & 0 \\
       \end{pmatrix}. \]
  Thus, $\sigma_1=\eta_0(\fq(A))$, see~\eqref{eta_0}.
\item\label{A3p3} \fbox{$A_t\neq\mathbf{0}$ and
  $B_t\neq\mathbf{0}$.} In this case,
  $n_{\fq}(S_t,n)=q^2n_{\fq}(S_{t-1},n-1)$ and we apply a recursive argument (see Case~B.\ref{B1} in \S~\ref{odd-t}).
  \end{enumerate}
\end{enumerate}
Recall that   $A_t=\mathbf{0}$ and
$B_t\neq\mathbf{0}$ cannot occur as the matrix $S_t$ is in $t$-REF.

Computing the value of $n_{\fq}(S_t,n)$ when
$(\alpha_t,\beta_t)=(0,0)$ has thus been reduced to
 determining $n_{\fq}(S_{t-1},n-1)$, where the
number of columns of $S_{t-1}$ is odd and the number of unknowns is $2n-t$.

\subsubsection{Odd $t$}
\label{odd-t}
For $t$ odd System~\eqref{mains} can be explicitly written as follows.
\begin{equation}
\setlength{\nulldelimiterspace}{0pt}
\label{e3}
\begin{cases}
{\displaystyle
 \alpha_1^2 +\!\!\!\sum_{i=1}^{(t-1)/2}\!\!\alpha_{2i}\alpha_{2i+1}+\!\!\!\!\!\sum_{i=(t+1)/2}^{n}\!\!\!\!x_{2i}x_{2i+1}=0}\\
\displaystyle\beta_1^2 +\!\!\!\sum_{i=1}^{(t-1)/2}\!\!\beta_{2i}\beta_{2i+1}+\!\!\!\!\!\sum_{i=(t+1)/2}^{n}\!\!\!y_{2i}y_{2i+1}=0\\
2\alpha_1\beta_1+\!\!\!{\displaystyle\sum_{i=1}^{(t-1)/2}}\!\!(\alpha_{2i}\beta_{2i+1} +\alpha_{2i+1}\beta_{2i}) +
\!\!\!{\displaystyle \sum_{i=(t+1)/2}^{n}}\!\!(x_{2i}y_{2i+1} +x_{2i+1}y_{2i})=0.\\
  \end{cases}
\end{equation}
As in Section~\ref{sec2.1}, let
$S_t=\begin{pmatrix}
A_t \\
B_t\\
  \end{pmatrix}$
be in RREF with $A_t=(\alpha_1,\ldots,\alpha_t)$ and $B_t=(\beta_1,\ldots,\beta_t)$.
We first replace the $(2\times t)$-matrix $S_t$ with the $2\times (t+1)$-matrix $S_{\gamma,\delta}$
obtained from $S_t$ by adding the column $\begin{pmatrix}\gamma\\ \delta\end{pmatrix}$, with $\gamma,\delta\in\FF_q$, i.e. $S_{\gamma,\delta}=\begin{pmatrix}
   A_t & \gamma \\
  B_t & \delta
  \end{pmatrix}.$
By Definition~\ref{nqns}, $n_{\fq}(S_{\gamma,\delta},n)=0$ if $S_{\gamma,\delta}$ is not in RREF. Hence
\[ n_{\fq}(S_t,n)=\sum_{(\gamma,\delta)\in\FF_q^2}n_{\fq}(S_{\gamma,\delta},n). \]
We distinguish several cases.
\begin{enumerate}[\mbox{B.}1)]
\item\label{B1}
\fbox{$A_t\neq\mathbf{0}$ and $B_t\neq\mathbf{0}$.}
We need to compute the values $n_{\fq}(S_{\gamma,\delta},n)$  where $S_{\gamma,\delta}$ has an even number  $t+1$ of columns. More precisely,
\[ n_{\fq}(S_t,n)=\!\!\!\!\sum_{\tiny{\begin{array}{l}
(\gamma, \delta)\in \FF_q^2\\
\gamma\not=0\not=\delta
\end{array}}
}\!\!\!\!\!\!n_{\fq}(S_{\gamma, \delta}, n)+\!\!\!\! \sum_{\tiny{
\gamma \in \FF_q\setminus \{0\}
}}\!\!\!\!n_{\fq}(S_{\gamma, 0}, n)+\!\!\!\! \sum_{\tiny{
\delta\in \FF_q\setminus \{0\}}
}\!\!\!\!n_{\fq}(S_{0, \delta}, n)+ n_{\fq}(S_{0, 0}, n).\]


\begin{enumerate}[\mbox{B.\ref{B1}.}1)]
\renewcommand{\theenumi}{\relax}
  \item\fbox{{$\displaystyle{\!\!\!\!\sum_{\tiny{\begin{array}{l}
(\gamma, \delta)\in \FF_q^2\\
\gamma\not=0\not=\delta
\end{array}}
}\!\!\!\!n_{\fq}(S_{\gamma, \delta}, n).}$}}
   Put
  $\lambda=\delta^{-1}\gamma$. We have
    $n_{\fq}(S_{\gamma,\delta},n)=n_{\fq}(S_{t+1},n)$ where
   \[S_{t+1}:=\begin{pmatrix}
      \alpha_1-\lambda\beta_1 & \ldots & \alpha_t-\lambda\beta_t & 0 \\
      \beta_1 & \ldots & \beta_{t} & \delta \\
      \end{pmatrix}.\]

    As $S_{t+1}$ is in $(t+1)$-REF, we are lead back to
    Case~A.\ref{A1q} of \S~\ref{even-t}. Thus,
      \[n_{\fq}(S_{t+1},n)=q^{2n-t-1} \eta_1(c),\]
      with $\eta_1(c)$ the number of solutions of the equation
      $\fq(A-\lambda B)=0$, as $\lambda$ varies in $\FF_q\setminus\{0\}$.
      If $c:=\fq(A-\lambda B)$, then (see also~\eqref{eta_0})
      \begin{equation}\label{eta_1}
     \eta_{1}(c):=\begin{cases}
 (q-1)|Q^+(2n-t-2,q)|_1+1\hskip-1cm\phantom{,}&
         \text{for $c=0$} \\
         |\PG(2n-t-2,q)|_1-|Q^+(2n-t-2,q)|_1&
         \text{for $c\neq 0$}, \\
       \end{cases}
     \end{equation}
 that is, by Table~\ref{unum}:
\[
  \eta_1(c)=\begin{cases}
(q-1)\cdot\frac{(q^{n-(t+1)/2}-1)(q^{n-(t+3)/2}+1)}{q-1}+1 & \text{if $c=0$} \\[6pt]
2\cdot\frac{1}{2}\left(\frac{q^{2n-t-1}-1}{q-1}-\frac{(q^{n-(t+1)/2}-1)(q^{n-(t+3)/2}+1)}{q-1}\right) &
\text{if $c\neq 0$}.\\
\end{cases}
\]

   We have
        \begin{equation}
          \label{eqxi}
          c=\fq(A)-\lambda\fb(A,B)+\lambda^2\fq(B).
        \end{equation}
        Let now $\xi$ be the number of non-zero solutions
        of~\eqref{eqxi} in the unknown $\lambda$.
        The possible values assumed by $\xi$ are
        outlined in Table~\ref{xiodd}
        for $q$ odd and in Table~\ref{xieven} for $q$ even.
        For $q$ odd, the symbols $\square$ and $\not\!\square$
        represent respectively the set of all non-zero square elements
        and the set of non-square elements in $\FF_q$.
\begin{table}[h]
\begin{minipage}[t]{5.5cm}
\caption{Number $\xi$ of solutions of  $\fq(A)-\lambda\fb(A,B)+\lambda^2\fq(B)=0$ for
$q$ odd.}
\[
\begin{array}{l|l|l|l||l}
  \fq(A) & \fb(A,B) & \fq(B) & \Delta  & \xi \\ \hline
  0     &  0   &   0   &    0    & q-1 \\
  0     &\neq0 & \neq0 &    *    & 1   \\
  0     &  0   & \neq0 &    0    & 0   \\
  0     &\neq0 &   0   &    *    & 0   \\
  \neq0 & 0    &   0   &    *    & 0   \\
  \neq0 &\neq0 &   0   &    *    & 1   \\
  \neq0 & \text{Any}   & \neq0 & \square & 2   \\
  \neq0 & \text{Any}   & \neq0 &     0   & 1   \\
  \neq0 & \text{Any}   & \neq0 & \not\!\square\, & 0
\end{array}
\]
\[\Delta:=\fb(A,B)^2-4\fq(A)\fq(B). \]
\centerline{ $*$ means that there are no conditions on $\Delta.$}
\label{xiodd}
\end{minipage}\hfill
\begin{minipage}[t]{5.5cm}
\caption{Number $\xi$ of solutions of \hfill\  $\fq(A)+\lambda\fb(A,B)+\lambda^2\fq(B)=0$ for $q$ even.}
\[
\begin{array}{l|l|l|l||l}
  \fq(A)  & \fb(A,B) & \fq(B) &   \Theta  & \xi \\ \hline
  0     &  0   &   0   &    *    & q-1 \\
  0     &\neq0 & \neq0 &    *    & 1   \\
  0     & 0    & \neq0 &    *    & 0   \\
  0     &\neq0 &   0   &    *    & 0   \\
  \neq0 &    0 &   0   &    *   & 0   \\
  \neq0 &\neq0 &   0   &    *    & 1   \\
  \neq0 & 0    & \neq0 &    *    & 1   \\
  \neq0 & \neq0   & \neq0 & 0  & 2   \\
  \neq0 & \neq0   & \neq0 & 1  & 0   \\
\end{array}
\]
\[\Theta:=\mathrm{Tr}_2\left(\frac{\fq(A)\fq(B)}{\fb(A,B)^2}\right). \]
\begin{centering}
 $*$ means that there are no conditions on $\Theta$ or $\Theta$ does not exist.
\end{centering}
\label{xieven}
\end{minipage}
\end{table}
Hence,
       \[  \sum_{\tiny{\begin{array}{l}
(\gamma, \delta)\in \FF_q^2\\
\gamma\not=0\not=\delta
\end{array}}
}n_{\fq}(S_{\gamma, \delta}, n)=
        \underbrace{(q-1)}_{\text{cases for } \delta}q^{2n-t-1}
         \big(\!\!\!\!\!
         \underbrace{\xi \eta_{1}(0)}_{\begin{subarray}{c}
             \text{ first eq.} \\ \text{homogeneous}
           \end{subarray}}+
           \underbrace{(q-1-\xi)\eta_{1}(1)}_{
             \begin{subarray}{c}\text{first eq.} \\
               \text{nonhomogeneous}\end{subarray}}
           \big).\]
  \item\fbox{$\displaystyle{\sum_{\tiny{
\gamma \in \FF_q\setminus \{0\}
}}n_{\fq}(S_{\gamma, 0}, n)+\sum_{\tiny{
\delta\in \FF_q\setminus \{0\}
}}n_{\fq}(S_{0, \delta}, n)}$.}
 Suppose $\gamma=0$ and $\delta\not=0.$ Arguing as in Case A.\ref{A1q} of
 \S~\ref{even-t}, we see that
 $n_{\fq}(S_{0,\delta},n)=n_{\fq}(S_{0,1},n)$ for any $\delta\neq 0$.
 Hence,
\[ \sum_{
\gamma \in \FF_q\setminus \{0\}
}n_{\fq}(S_{\gamma, 0}, n)=(q-1)n_{\fq}(S_{0,1},n), \]
where the factor $n_{\fq}(S_{0,1},n)$ can be directly computed as in Case A.\ref{A1q}.

The case $\gamma\neq0$ and $\delta=0$ is analogous to case A.2.\ref{A21};
hence,
\[ \sum_{
\delta\in \FF_q\setminus \{0\}
}n_{\fq}(S_{0, \delta}, n)=(q-1)n_{\fq}(S_{1,0},n). \]

\item\fbox{$n_{\fq}(S_{0,0},n)$.} In this case,
  $n_{\fq}(S_{0,0},n)=q^2n_{\fq}(S_t,n-1)$ as
  $x_{t+2}$ and $y_{t+2}$ may be chosen arbitrarily and $n_{\fq}(S_t,n-1)$ is the number of
  solutions of the system associated with
  \[ G'=\begin{pmatrix}
   \alpha_1 & \alpha_2 & \ldots & \alpha_t & x_{t+3} & \ldots & x_{2n+1} \\
   \beta_1 & \beta_2 & \ldots & \beta_t & y_{t+3} & \ldots & y_{2n+1}
    \end{pmatrix}. \]
  \end{enumerate}
\item \fbox{$A_t=B_t=\mathbf{0}$.}\label{B2}
  An argument analogous to Case~A.\ref{A3q}.\ref{A3p1} of
  \S~\ref{even-t}  shows that
  we have to determine the number
  of lines of a hyperbolic quadric $\cQ^+$ in $\PG(2n-t,q)$ with equation
  \[ x_{t+1}x_{t+2}+\ldots+x_{2n}x_{2n+1}=0; \]
  we refer to Table~\ref{unum} for the actual value.
\item \fbox{$A_t\neq\mathbf{0}$ and $B_t=\mathbf{0}$.} \label{B3}
In this case all the $2\times (t+1)$-matrices $S_{\gamma,\delta}$ are of the form
  \[ S_{\gamma,\delta}=\begin{pmatrix}
    \alpha_1 & \alpha_2 & \ldots & \alpha_t & \gamma \\
    0 & 0 & \ldots & 0 & \delta \\
    \end{pmatrix}. \]
When $S_{\gamma, \delta}$ is taken in $(t+1)$-REF, either $\delta=1$ and $\gamma=0$ or
    $\delta=0$ and $\gamma$ is arbitrary.
    Note that for any solution of the first equation of \eqref{e3} there
    are $q-1$ vector solutions of the second equation yielding the same line.
In this case,
\[ n_{\fq}(S_t,n)=  n_{\fq}(S_{0, 1}, n)+\!\!\! \sum_{
\gamma \in \FF_q\setminus \{0\}
}\!\!\!n_{\fq}(S_{\gamma, 0}, n)+ n_{\fq}(S_{0, 0}, n).\]

  \begin{enumerate}[\mbox{B.\ref{B3}.}1)]
\renewcommand{\theenumi}{\relax}
    \item\fbox{$n_{\fq}(S_{0, 1}, n).$}
      We can compute $n_{\fq}(S_{0,1},n)$ using the same approach as in Case~A.\ref{A1q}.
  \item\fbox{$\displaystyle{\sum_{\tiny{
\gamma \in \FF_q\setminus \{0\}
}}\!\!\!n_{\fq}(S_{\gamma, 0}, n).}$} We have
  $S_{\gamma,0}=\begin{pmatrix}
    \alpha_1 & \alpha_2 & \ldots & \alpha_t & \gamma \\
    0 & 0 & \ldots & 0 & 0
    \end{pmatrix}.$
   With an approach similar to Case~A.\ref{A2q}.\ref{A22} we see that
    \[ n_{\fq}(S_{\gamma,0},n)=n_{\fq}(S_{1,0},n)=q^{2n-t-2}\cdot\frac{(q^{n-(t+1)/2}-1)(q^{n-(t+3)/2}+1)}{q-1}. \]
    Hence, this case contributes $q^{2n-t-2}(q^{n-(t+1)/2}-1)(q^{n-(t+3)/2}+1)$ to $n_{\fq}(S_t,n)$.

\item\fbox{$n_{\fq}(S_{0, 0}, n)$.}\label{B33}
 We need to compute $n_{\fq}(S_{0,0},n)$, i.e. the number of solutions of the following system
 in the unknowns $x_{t+2},\ldots,x_{2n+1},y_{t+2},\ldots,y_{2n+1}$:
          \begin{equation}\label{S9p} \begin{cases}
            \fq(A)+0x_{t+2}+x_{t+3}x_{t+4}+\ldots+x_{2n}x_{2n+1} = 0 \\
            0y_{t+2}+y_{t+3}y_{t+4}+\ldots+y_{2n}y_{2n+1}=0 \\
            0y_{t+2}+x_{t+2}0+x_{t+3}y_{t+4}+x_{t+4}y_{t+3}+\ldots+x_{2n+1}y_{2n}=0.
      \end{cases}
      \end{equation}
      We shall refer to the system in the unknowns $x_{t+3},\ldots,x_{2n+1},y_{t+3},\ldots,y_{2n+1}$ obtained
      from~\eqref{S9p} by removing the unknowns $x_{t+2}$ and $y_{t+2}$ as the  ``reduced'' one.
      With arguments similar to those of  Case~A.\ref{A3q}.\ref{A3p2} we see that each solution of the reduced system
      corresponds to $q^2$ solutions of~\eqref{S9p}.
      However, there are also solutions of~\eqref{S9p} not arising from the reduced system.
      These solutions correspond to cases in which $y_{t+2}\neq0$ and $y_{t+3}=\cdots=y_{2n+1}=0$; that is, the representative matrix of the totally singular lines considered in these cases is
      \[ \begin{pmatrix}
       \alpha_1 & \alpha_2 & \ldots & \alpha_t & 0 & 0 & x_{t+3} & \ldots & x_{2n+1} \\
       0 & 0 & \ldots & 0 & 0 & 1 & 0 & \ldots & 0
       \end{pmatrix}; \]
       there are $\eta_1(\fq(A))$ possibilities; see~\eqref{eta_1}.
       Hence,
  \[ n_\fq(S_{0,0},n)=q^2n_{\fq}(S_t,n-1)+\eta_1(\fq(A)). \]

\end{enumerate}
\end{enumerate}
 The case $A=\mathbf{0}$ and
  $B\not=\mathbf{0}$ cannot happen according to the convention we adopted.

The above arguments provide a complete description of how to compute
the function $n_{\fq}(S_t,n)$ for any $S_t\in\fR^{1}_2$ and $n\in\NN$.
We summarize the details of the algorithm in Table~\ref{tabQ}.
\begin{table}
\caption{Enumerator for Orthogonal Line Grassmannians}
\[ \begin{array}{c|c|c|c|c}
S_t=\begin{pmatrix} A_t\\ B_t\end{pmatrix} & {t} & \text{Case} & n_\fq(S,n) & \text{Complexity} \\ \hline
\emptyset & 0 & &
\begin{array}[c]{c}
\\
\frac{(q^n-1)(q^{n+1}-1)}{(q^2-1)(q-1)}\\
\\
\end{array}  & O(1)
 \\ \hline
\begin{array}[c]{c}
\begin{pmatrix}
\alpha_1 & \alpha_2 & \ldots & \alpha_{t-1} & 0 \\
\beta_1 & \beta_2 & \ldots & \beta_{t-1} & \beta_{t}
\end{pmatrix} \\
\beta_t\neq0
\end{array} & \text{Even} & A.\ref{A1q} & q^{2n-t}\eta_0(\fq(A)) & O(t) \\ \hline
\begin{array}[c]{c}
\begin{pmatrix}
\alpha_1 & \alpha_2 & \ldots & \alpha_{t-1} & \alpha_t \\
\beta_1 & \beta_2 & \ldots & \beta_{t-1} & 0
\end{pmatrix} \\
\alpha_t\neq0 \\
(\beta_1,\ldots,\beta_{t-1})\neq\mathbf{0}
\end{array} & \text{Even} &  A.\ref{A2q}.\ref{A21} & q^{2n-t}\eta_0(\fq(B)) & O(t) \\ \hline
\begin{array}[c]{c}
\begin{pmatrix}
\alpha_1 & \alpha_2 & \ldots & \alpha_{t-1} & \alpha_t \\
0 & 0 & \ldots & 0 & 0
\end{pmatrix} \\
\alpha_t\neq0 \\
\end{array} & \text{Even} &  A.\ref{A2q}.\ref{A22} & \frac{q^{2n-t-1}}{q-1}(q^{n-t/2}-1)(q^{n-t/2-1}+1) & O(1) \\ \hline
\begin{array}[c]{c}
\begin{pmatrix}
0 & 0 & \ldots & 0 & 0 \\
0 & 0 & \ldots & 0 & 0
\end{pmatrix} \\
\end{array} & \text{Even} & A.\ref{A3q}.\ref{A3p1} & \sigma q^2+(q^{n-t/2}-1)(q^{n-t/2-1}+1) & O(1) \\ \hline
\begin{array}[c]{c}
\begin{pmatrix}
\alpha_1 & \alpha_2 & \ldots & \alpha_{t-1} & 0 \\
0 & 0 & \ldots & 0 & 0
\end{pmatrix} \\
(\alpha_1,\ldots,\alpha_{t-1})\neq\mathbf{0} \\
\end{array} & \text{Even} & A.\ref{A3q}.\ref{A3p2} &
\begin{array}[c]{c}
q^2n_{\fq}^O(S_{t-1},n-1)+\eta_0(\fq(A));\\
\text{\rm with }n_{\fq}^O(S_{t-1},n-1)\ \text{\rm as in Case }B.\ref{B3}\\
\end{array}
& O(n^2) \\ \hline
\begin{array}[c]{c}
\begin{pmatrix}
\alpha_1 & \alpha_2 & \ldots & \alpha_{t-1} & 0 \\
\beta_1 & \beta_2 & \ldots & \beta_{t-1} & 0
\end{pmatrix} \\
(\alpha_1,\ldots,\alpha_{t-1})\neq\mathbf{0} \\
(\beta_1,\ldots,\beta_{t-1})\neq\mathbf{0} \\
\end{array} & \text{Even} & A.\ref{A3q}.\ref{A3p3} &
\begin{array}[c]{c}
q^2n_{\fq}^O(S_{t-1},n-1);\\
\text{\rm with }n_{\fq}^O(S_{t-1},n-1)\ \text{\rm as in Case }B.\ref{B1}\\
\end{array}
 & O(n^2) \\ \hline
\begin{array}[c]{c}
\begin{pmatrix}
\alpha_1 & \alpha_2 & \ldots & \alpha_{t-1} & \alpha_t \\
\beta_1 & \beta_2 & \ldots & \beta_{t-1} & \beta_t
\end{pmatrix} \\
(\alpha_1,\ldots,\alpha_{t-1},\alpha_t)\neq\mathbf{0} \\
(\beta_1,\ldots,\beta_{t-1},\beta_t)\neq\mathbf{0} \\
\end{array} & \text{Odd} & B.\ref{B1} &
\begin{array}[c]{c}
\\
q^{2n-t-1}\eta_1(c)+(q-1)n_{\fq}^E(S_{0,1},n)+\\
(q-1)n_{\fq}^E(S_{1,0},n))+ q^2n_{\fq}^O(S_t,n-1);\\
\text{\rm with } n_{\fq}^E(S_{0,1},n)\ \text{\rm as  in Case }A.\ref{A1q}\,\text{\rm and}\\
n_{\fq}^E(S_{1,0},n)\ \text{\rm as in Case }A.\ref{A2q}.\ref{A1q}.\\
\\
\end{array}
& O(n^2) \\ \hline
\begin{array}[c]{c}
\begin{pmatrix}
0 & 0 & \ldots & 0 & 0 \\
0 & 0 & \ldots & 0 & 0
\end{pmatrix} \\
\end{array} & \text{Odd} & B.\ref{B2}  &
\frac{(q^{2n-t-1}-1)(q^{n-(t-1)/2}-1)(q^{n-(t-3)/2}+1)}{(q^2-1)(q-1)}
  & O(1) \\ \hline
\begin{array}[c]{c}
\begin{pmatrix}
\alpha_1 & \alpha_2 & \ldots & \alpha_{t}  \\
0 & 0 & \ldots & 0
\end{pmatrix} \\
(\alpha_1,\ldots,\alpha_{t})\neq\mathbf{0} \\
\end{array} & \text{Odd} & B.\ref{B3} &
\begin{array}[c]{c}
\\
n_{\fq}^E(S_{0,1},n)+n_{\fq}^E(S_{1,0},n)+\\
q^2n_{\fq}^E(S_t,n-1)+\eta_1(\fq(A));\\
{\rm with\, } n_{\fq}^E(S_{0,1},n) {\rm\,as\,\, in\,\,Case\,\,}A.\ref{A1q}\,{\rm and}\\
n_{\fq}^E(S_{1,0},n) {\rm\,as\,\, in\,\,Case\,\,}A.\ref{A2q}.\ref{A2q}.\\
\\
 \end{array}
 & O(n^2) \\ \hline
     \end{array} \]
\label{tabQ}
\centerline{
For the meaning of the symbols and the constants, see the relevant Cases A.x and B.x. in \S~\ref{even-t} and \S~\ref{odd-t}.}
\end{table}

\subsubsection{Complexity}
\label{s4}
We now analyze the complexity of the algorithm described in \S~\ref{even-t} and \S~\ref{odd-t}.

Given a $(2\times t)$-matrix $S_t=\begin{pmatrix} A_t\\B_t\end{pmatrix}$ in RREF and $n\in\NN$,
we shall denote by  $n^E_{\fq}(S_t,n)$  the output of the algorithm
with $t$ even (see \S~\ref{even-t}) and by  $n^O_{\fq}(S_t,n)$ the output with $t$ odd (see
\S~\ref{odd-t}).
We will write $\kappa(n_{\fq}(S_t,n))$ for the number of multiplications required to compute
$n_{\fq}(S_t,n)$.
The complexity of the various steps of the algorithm will now be  examined.

\noindent {\sc Step 1.} If $S_0=\emptyset$ or we are in the hypotheses of cases A.\ref{A2q}.\ref{A22}, A.\ref{A3q}.\ref{A3p1} or
 B.\ref{B2}, then we can provide the value of $n_{\fq}(S_t,n)$ by directly applying a formula with
fixed complexity $O(1)$;
 see also Table~\ref{tabQ}.
 \par
\noindent {\sc Step 2.} Otherwise, transform $S_t$ in $t$-REF; this requires $t$ products and $t$ sums. \par
\noindent {\sc Step 3.} If $t$ is even, compute $n_{\fq}^E(S_t,n)$; otherwise compute $n_{\fq}^O(S_t,n)$. \par

Clearly,
 \[ \kappa(n_{\fq}(S_t,n))\leq t+\max\{\kappa(n_{\fq}^E(S_t,n)),\kappa(n_{\fq}^O(S_t,n))\}. \]
 We shall now analyze in detail $\kappa(n_{\fq}^E(S_t,n))$ and $\kappa(n_{\fq}^O(S_t,n))$.

\begin{itemize}
    \item {\fbox{$n_{\fq}^E(S_t,n)$.}}
        \begin{enumerate}
         \item If $S_t$ is as in Case A.\ref{A1q}, then we need to evaluate $\fq(A)$; this
         requires $t$ products and $t$ sums; thus it has complexity $O(t)$. Likewise, under the
         assumptions of A.\ref{A2q}.\ref{A21}, the complexity to determine $n_{\fq}(S_t,n)$ is $O(t)$.
         \item If $S_t$ satisfies the hypotheses of A.\ref{A3q}.\ref{A3p2} or A.\ref{A3q}.\ref{A3p3}, we need
         to consider the complexity of $n_{\fq}^O(S_{t-1},n-1)$ where $S_{t-1}$ is a $2\times (t-1)$-matrix obtained
         from $S_t$ by deleting its last column. Then we need to consider
         a case of odd length $n_{\fq}^O(S_{t-1},n-1)$. As it will be shown below, the complexity here is
         at most $O(n^2)$.
         \end{enumerate}
    \item {\fbox{$n_{\fq}^O(S_t,n)$.}}
    We claim that
    \begin{equation}
     \label{rece}
     \kappa(n_{\fq}^O(S_t,n))\leq 3t+\kappa(n_{\fq}^O(S_t,n-1)).
     \end{equation}
    By Cases B.\ref{B1} and B.\ref{B3}, computing $n_{\fq}^O(S_t,n)$ requires to determine the values of
    $n_{\fq}^E(S_{\gamma,\delta},n)$ for $(\gamma,\delta)=(0,1)$, $(\gamma,\delta)=(1,0)$ and
    $\gamma\neq0\neq\delta$ (Case B.\ref{B1})) and the value of $n_{\fq}^O(S_t,n-1)$ (Case B.\ref{B3}). The first three cases have already been
    shown to have complexity at most $O(t)$. Hence, the claim follows.
    \par
    Observe that $\kappa(n_{\fq}^O(S_t,\frac{t-1}{2}))=3\frac{t-1}{2}$, since, in this case, we just need to
    check if the line spanned by $A$ and $B$ is totally singular. Note that $n_{\fq}^O(S_t,\frac{t-1}{2})$ is
    the number of totally singular lines of $\PG(t-1,q)$ whose representative matrix is $S_t.$
    Clearly, this number is $1$ if the line spanned by the rows of $S_t$ is singular and $0$ otherwise.
    By recursively applying~\eqref{rece}, since $t\leq 2n+1$, we have
    \begin{multline*}
      \kappa(n_{\fq}^O(S_t,n))\leq 3t+\kappa(n_{\fq}^O(S_t,n-1))+O(1)
      \leq 6t+\kappa(n_{\fq}^O(S_t,n-2))+O(1)\leq\ldots \\
      \leq 3\sum_{i=1}^{n-(t+1)/2}t+\kappa(n_{\fq}^O(S_t,\frac{t-1}{2}))+O(1)\leq O(n^2).
    \end{multline*}
\end{itemize}
    In summary, the complexity of the algorithm to determine $n_{\fq}(S_t,n)$ is $O(n^2)$. This proves Theorem~\ref{t15} for the orthogonal line Grassmannian.

\subsection{Enumerating Symplectic Grassmannians}
\label{w2}
Following Definition~\ref{nqns}, given any $S_t=\begin{pmatrix} A_t\\B_t\end{pmatrix}\in\fR_{2}^0$,
where $A_t=(\alpha_1,\ldots,\alpha_t)$, $B_t=(\beta_1,\ldots,\beta_t)$, $1\leq t\leq 2n,$
denote by
$n_{\fs}(S_t,n)$  the number of totally
$\fs$-isotropic lines of $\PG(2n-1,q)$ spanned by $\widehat{A}=(\alpha_1,\ldots,\alpha_t,x_{t+1},\ldots,x_{2n})$
and $\widehat{B}=(\beta_1,\ldots,\beta_t,y_{t+1},\ldots,y_{2n})$ as $x_{t+1},\ldots,x_{2n},y_{t+1},\ldots,y_{2n}$
vary.
Also, put $A=(A_t,0,\ldots, 0)$ and $B=(B_t,0,\ldots, 0)$. We remind that $S_t$
is assumed to be in RREF. (Otherwise, $n_{\fs}(S_t,n)=0$).
In this section we compute $n_{\fs}(S_t,n)$, with an approach similar to that of Section~\ref{s2}.
We have to determine the number of solutions of the equation $\fs(\hA,\hB)=0$ in the unknowns
$x_{i},y_{i}$ for $t+1\leq i\leq 2n$, see Definition~\ref{t-REF}.
The first step of the algorithm is to transform $S_t$ in $t$-REF using~\eqref{SpC}.

Let $\fs'$ be the alternating
form induced by the restriction of $\fs$ to the subspace of $\overline{V}$ of equation
$x_1=x_2=\cdots=x_t=0$. We distinguish two subcases.
\subsubsection{Even $t$}\label{Sym-even-t}
There are three possibilities:
\begin{enumerate}[\mbox{C.}1)]
\item \framebox{$A_t=B_t={\mathbf{0}}$.}\label{C1}
  In this case, $n_{\fs}(S_t,n)$ is the number of totally $\fs'$-isotropic lines
  in a subspace $\PG(2n-t-1,q).$
  Thus (see Table~\ref{unum}),
  \[ n_{\fs}(S_t,n)=\frac{(q^{2n-t}-1)(q^{2n-t-2}-1)}{(q-1)(q^2-1)}. \]
  \item \framebox{$A_t\neq\mathbf{0}$, $B_t=\mathbf{0}$.}\label{C2}
    In this case, the representative matrix of the lines we consider has the form
     \[G=\begin{pmatrix} \alpha_1 & \ldots & \alpha_t & x_{t+1} & \ldots & x_{2n} \\
                                          0   & \ldots & 0 & y_{t+1} & \ldots & y_{2n}
                                          \end{pmatrix}. \]
     Suppose $Y=(y_{t+1},\ldots,y_{2n})$ is a given non-null vector with leading coefficient
     $y_i=1$, $i>t$.
     There are $\frac{q^{2n-t}-1}{q-1}$ choices for $Y$.
     For any such $Y$ we count the number of vectors $X=(x_{t+1},\ldots,x_{2n})$ with
     $x_i=0$ such that $\fs'(X,Y)=0$: this amounts to $q^{2n-t-2}$.
     Hence,
     \[ n_{\fs}(S_t,n)=q^{2n-t-2}\frac{q^{2n-t}-1}{q-1}. \]
  \item \framebox{$A_t\neq\mathbf{0}$, $B_t\neq\mathbf{0}$.}\label{C3}
  We distinguish
    two subcases, according to the value of $\fs(A,B)$.
    \begin{enumerate}[\mbox{C.\ref{C3}}.1)]
    \renewcommand{\theenumi}{\relax}
      \item \framebox{$\fs(A,B)=0$.}\label{C31}
        In this case we count the number of
        pairs of vectors $(X,Y)$ with
        $X,Y\in \FF_q^{2n-t}$ and $\fs'(X,Y)=0$.
        If $X={\mathbf0}$, then there are $q^{2n-t}$ different choices
        for $Y$ such that $\fs'(X,Y)=0$. If $X\neq{\mathbf0}$,
        there are $q^{2n-t-1}$ choices for $Y$ such that $\fs'(X,Y)=0$.
        Thus,
        \begin{equation}
          \label{eA0}
          n_{\fs}(S_t,n)=q^{2n-t-1}(q^{2n-t}-1)+q^{2n-t}.
        \end{equation}
      \item \framebox{$\fs(A,B)\neq 0$.}\label{C32}
        Let $X=(x_{t+1},\ldots,x_{2n})$ be a fixed non-null vector.
        There are $q^{2n-t}-1$ choices for such $X$.
        We count the number of vectors $Y=(y_{t+1},\ldots,y_{2n})$
        such that
        $\fs'(X,Y)=-\fs(A,B)$.
        This is a linear equation in the unknowns $y_{t+1},\ldots,y_{2n}$;
        hence, there are $q^{2n-t-1}$ choices for $Y$. Thus,
        \begin{equation}
          \label{eA1}
          n_{\fs}(S_t,n)=(q^{2n-t}-1)q^{2n-t-1}.
        \end{equation}
      \end{enumerate}
\end{enumerate}
\subsubsection{Odd $t$}
\label{Sym-odd-t}
When $t$ is odd, we first replace the matrix
$S_t=\begin{pmatrix}
    A_t\\
   B_t
    \end{pmatrix}\in \fR_a^0$, with a $2\times (t+1)$-matrix
    $S_{\gamma,\delta}:= \begin{pmatrix}
    A_t & \gamma \\
    B_t & \delta
    \end{pmatrix}$
    obtained from $S_t$ by adding  the column
    $\begin{pmatrix}\gamma\\ \delta\end{pmatrix}$, with $\gamma, \delta\in \FF_q$.
    By Definition~\ref{nqns}, $n_{\fs}(S_{\gamma,\delta},n)=0$ if $S_{\gamma,\delta}$ is not in RREF. Hence
\[ n_{\fs}(S_t,n)=\sum_{(\gamma,\delta)\in\FF_q^2}n_{\fs}(S_{\gamma,\delta},n). \]
   We distinguish three subcases.
\begin{enumerate}[\mbox{D.}1)]
  \item \framebox{$A_t=B_t={\mathbf{0}}$.}\label{D1}
 In this case,
 $n_{\fs}(S_t,n)$ is the number of lines contained in the symplectic polar space  $\cW$ defined
 by $\fs$ and in
 the subspace $\Pi$ of codimension $t$ described by the equations
 \[ x_1=0, x_2=0, \ldots, x_{t}=0. \]
As
 $\Pi\cap\cW$ is a degenerate symplectic polar space with radical of dimension $1$,
we have $n_{\fs}(S_t,n)=\sigma q^2+|\cW'|_1$, where
$\cW'$ is a non-degenerate symplectic polar space
in $\PG(2n-t-2,q)$ and
$|\cW'|_1$ and $\sigma:=|\cW'|_2$
are respectively the number of points and lines of $\cW'$;
see Table~\ref{unum}.
\item \framebox{$A_t\neq\mathbf{0}$, $B_t=\mathbf{0}$.}\label{D2}
  In this case the only matrices in RREF are
  $S_{\gamma,0}$ with $\gamma\in\FF_q$ and $S_{0,1}$.
  Thus,
  \[ n_{\fs}(S_t,n)=\sum_{\gamma\in\FF_q}n_{\fs}(S_{\gamma,0},n)+
    n_{\fs}(S_{0,1},n). \]
       By Case C.\ref{C2} of \S~\ref{Sym-even-t},
       $n_{\fs}(S_{\gamma,0},n)=n_{\fs}(S_{0,0},n)$ for all $\gamma\in\FF_q$; thus,
       \[ n_{\fs}(S_t,n)=qn_{\fs}(S_{0,0},n)+n_{\fs}(S_{0,1},n), \]
       where $n_{\fs}(S_{0,1},n)$ is computed in Case~C.\ref{C3} of \S~\ref{Sym-even-t}
       (and $n_{\fs}(S_{0,0},n)$ is computed in Case~C.\ref{C2}).
  \item \framebox{$A_t\neq\mathbf{0}, B_t\neq\mathbf{0}$.}\label{D3}
    There are two possibilities.
    \begin{enumerate}[\mbox{D.\ref{D3}.}1)]
      \renewcommand{\theenumi}{\relax}
      \item \framebox{$\alpha_t=\beta_t=0$.}\label{D31} In this case, the matrix $S_{\gamma, \delta}$ has the form
      \[ S_{\gamma,\delta}=\begin{pmatrix}
       \alpha_1 & \ldots & \alpha_{t-1} & 0 & \gamma \\
       \beta_1  & \ldots & \beta_{t-1} & 0 & \delta
       \end{pmatrix} \]
       and it is in $t$-REF (as $S_t$ is in $t$-REF).
       Observe that the number of lines admitting a representative matrix in $t$-REF whose $(t+1)$-prefix is $S_{\gamma,\delta}$ can be computed as in Case C.\ref{C3}.\ref{C31} or C.\ref{C3}.\ref{C32} of \S~\ref{Sym-even-t}
       according as $\fs(A,B)=0$ or $\fs(A,B)\neq0$,
       but does not depend on the choice of
       $\gamma$ and $\delta$.
        Thus,
        \[ n_{\fs}(S_t,n)=q^2n_{\fs}(S_{t-1},n). \]
      \item \framebox{$(\alpha_t,\beta_t)\neq(0,0)$.}\label{D32}
        Let $A_{\gamma}=(\alpha_1,\ldots,\alpha_t,\gamma,0,\ldots,0)$ and
        $B_{\delta}=(\beta_1,\ldots,\beta_t,\delta,0,\ldots,0)$.
        Clearly,
        \[ \fs(A_{\gamma},B_{\delta})=\alpha_1\beta_2-\alpha_2\beta_1+\cdots+
          \alpha_t\delta-\beta_t\gamma. \]
        As $(\alpha_i,\beta_i)$ for $i=1,\ldots,t$ are all given and $(\alpha_t,\beta_t)\neq(0,0)$,
        $\fs(A_{\gamma},B_{\delta})=0$ is a non-trivial linear equation in the unknowns $\gamma$ and
        $\delta$.
                    Hence, there are exactly
        $q$ values of $(\gamma,\delta)$ such that $\fs(A_{\gamma},B_{\delta})=0$.
        For each of these values we have, by Case~C.\ref{C3}.\ref{C31},
        $q^{2n-t-2}(q^{2n-t-1}-1)+q^{2n-t-1}$ distinct lines to take into account.
        For
        the remaining $q^2-q$ values of $(\gamma,\delta)$, such that $\fs(A_{\gamma},B_{\delta})\neq0$,
        we have, by Case~C.\ref{C3}.\ref{C32},
        $(q^{2n-t-1}-1)q^{2n-t-2}$ distinct lines.
        Consequently,
     \[ n_{\fs}(S_t,n)=q^{4n-2t-1}. \]
\end{enumerate}
\end{enumerate}
\subsubsection{Complexity}
Given a $(2\times t)$-matrix $S_t$ in RREF and $n\in\NN$, the computational complexity of the algorithm to determine
 $n_{\fs}(S_t,n)$ is $O(n)$.
 This can be immediately seen by analyzing the steps presented in the previous sections.
 For the convenience of the reader we summarize the various cases, depending on the structure of $S_t$, together with
 their complexity, in Table~\ref{tabS}.
 This proves Theorem~\ref{t15} for symplectic line Grassmannians.

\begin{table}
\caption{Enumerator for Symplectic Line Grassmannians}
\[ \begin{array}{c|c|c|c|c}
S=\begin{pmatrix} A_t\\ B_t\end{pmatrix} & {t} & \text{Case} & n_\fs(S,n) & \text{Complexity} \\ \hline
\emptyset & 0 & &
\begin{array}[c]{c}
\\
\frac{(q^{2n-1}-1)(q^{2n-2}-1)}{(q^2-1)(q-1)}\\
\\
\end{array}
 & O(1) \\ \hline
\begin{array}[c]{c}
\begin{pmatrix}
0 & 0 & \ldots & 0 & 0 \\
0 & 0 & \ldots & 0 & 0
\end{pmatrix} \\
\end{array} & \text{Even} & C.\ref{C1} &
\begin{array}[c]{c}
\\
\frac{(q^{2n-t}-1)(q^{2n-t-2}-1)}{(q^2-1)(q-1)}  \\
\\
\end{array}
& O(1) \\ \hline
\begin{array}[c]{c}
\begin{pmatrix}
\alpha_1 & \alpha_2 & \ldots & \alpha_{t-1} & \alpha_t \\
0 & 0 & \ldots & 0 & 0
\end{pmatrix} \\
(\alpha_1,\ldots,\alpha_t)\neq0 \\
\end{array} & \text{Even} &  C.\ref{C2} & q^{2n-t-2}\frac{q^{2n-t}-1}{q-1} & O(1) \\ \hline
\begin{array}[c]{c}
\begin{pmatrix}
\alpha_1 & \alpha_2 & \ldots & \alpha_{t}  \\
\beta_1 & \beta_2 & \ldots & \beta_{t}
\end{pmatrix} \\
(\alpha_1,\ldots,\alpha_t)\neq\mathbf{0} \\
(\beta_1,\ldots,\beta_t)\neq\mathbf{0} \\
\fs(A,B)=0
\end{array} & \text{Even} & C.\ref{C3}.\ref{C31} &
q^{2n-t-1}(q^{2n-t}-1)+q^{2n-t} & O(t) \\ \hline
\begin{array}[c]{c}
\begin{pmatrix}
\alpha_1 & \alpha_2 & \ldots & \alpha_{t}  \\
\beta_1 & \beta_2 & \ldots & \beta_{t}
\end{pmatrix} \\
(\alpha_1,\ldots,\alpha_t)\neq\mathbf{0} \\
(\beta_1,\ldots,\beta_t)\neq\mathbf{0} \\
\fs(A,B)\neq0
\end{array} & \text{Even} & C.\ref{C3}.\ref{C32} & (q^{2n-t}-1)q^{2n-t-1} & O(t) \\ \hline
\begin{array}[c]{c}
\begin{pmatrix}
0 & 0 & \ldots & 0 & 0 \\
0 & 0 & \ldots & 0 & 0
\end{pmatrix} \\
\end{array} & \text{Odd} & D.\ref{D1} &
\frac{(q^{2n-t-1}-1)^2}{(q-1)(q^2-1)}
   & O(1) \\ \hline
\begin{array}[c]{c}
\begin{pmatrix}
\alpha_1 & \alpha_2 & \ldots & \alpha_{t-1} & \alpha_t \\
0 & 0 & \ldots & 0 & 0
\end{pmatrix} \\
(\alpha_1,\ldots,\alpha_t)\neq0 \\
\end{array} & \text{Odd} &  D.\ref{D2} &
\begin{array}[c]{c}
\\
qn_{\fs}^E(S_{0,0},n)+n_{\fs}^E(S_{0,1},n)\\
\text{\rm with }n_{\fs}^E(S_{0,0},n)\ \text{\rm as in Case }C.\ref{C2}\ \text{\rm and}\\
n_{\fs}^E(S_{0,1},n)\ \text{\rm as in Cases }C.\ref{C3}.\ref{C31}\ \text{\rm or } C.\ref{C3}.\ref{C32}\\
\\
\end{array}
 & O(1) \\ \hline
\begin{array}[c]{c}
\begin{pmatrix}
\alpha_1 & \alpha_2 & \ldots & \alpha_{t-1} & 0  \\
\beta_1 & \beta_2 & \ldots & \beta_{t-1} & 0
\end{pmatrix} \\
(\alpha_1,\ldots,\alpha_{t-1})\neq\mathbf{0} \\
(\beta_1,\ldots,\beta_{t-1})\neq\mathbf{0} \\
\end{array} & \text{Odd} & D.\ref{D3}.\ref{D31} &
\begin{array}[c]{c}
\\
q^2n_{\fs}^E(S_{t-1},n)\,\,\,\text{\rm with }\\
n_{\fs}^E(S_{t-1},n)\ \text{\rm as in Cases }C.\ref{C3}.\ref{C31}\ \text{\rm or } C.\ref{C3}.\ref{C32}\\
\\
\end{array} & O(t) \\ \hline
\begin{array}[c]{c}
\begin{pmatrix}
\alpha_1 & \alpha_2 & \ldots & \alpha_{t}   \\
\beta_1 & \beta_2 & \ldots & \beta_{t}
\end{pmatrix} \\
(\alpha_1,\ldots,\alpha_{t-1})\neq\mathbf{0} \\
(\beta_1,\ldots,\beta_{t-1})\neq\mathbf{0} \\
(\alpha_t,\beta_t)\neq(0,0)
\end{array} & \text{Odd} & D.\ref{D3}.\ref{D32} & q^{4n-2t-1}  & O(1) \\ \hline

\end{array} \]
\label{tabS}
\noindent
\centerline{
For the meaning of the symbols and the constants, see the relevant Cases C.x and D.x. in
\S~\ref{Sym-even-t} and \S~\ref{Sym-odd-t}.}
\end{table}

\section{Enumerative coding}
\label{s3}
In this section, following the approach of~\cite{Cover}, we construct enumerators for the points of
$\Delta_{n,2}$ and
$\overline{\Delta}_{n,2}$
using the functions
$n_{\fq}$ and $n_{\fs}$ introduced in Section~\ref{p2}.
We shall present the full details for the orthogonal Grassmannian $\Delta_{n,2}$;
the symplectic case is entirely analogous.

Fix a total order $\preceq$ on the vectors of $\FF_q^2$ and
write $A\prec B$ if and only if $A\preceq B$ and $A\neq B$.
Let $\ell$ be a totally $\fq$-singular line of $V$ and $G_{\ell}=(G_1,\ldots,G_{2n+1})$ be its
$2\times (2n+1)$-representative matrix (in RREF), where $G_i\in\FF_q^2$ is the $i$-th column of $G_{\ell}$.
For any $j\leq 2n+1$ and $X\in\FF_q^2$, let
$S_{j}^X:=(G_1,\ldots,G_{j-1},X)$ be the $(2\times j)$-matrix comprising the first $j-1$ columns
of $G_{\ell}$ and whose last column is $X$.

Let $\iI=\{0,\dots, N-1\}$,
with $N=|\Delta_{n,2}|_1$ (see Table~\ref{unum}) and define
\begin{equation} \label{iota}
  \iota:\begin{cases}\Delta_{n,2}\to\iI \\
  \ell\mapsto \iota(G_{\ell}):=\displaystyle
  \sum_{j=1}^{2n+1}\sum_{X\prec G_j}n_{\fq}(S_j^X,n).
\end{cases}
\end{equation}
The order $\ors$ defined on the vectors of $\FF_q^2$ can be extended to
 matrices of order $2\times (2n+1)$ lexicographically; that is  $G\ll H$
if and only if there exists
$i\in \{1,\dots, 2n+1\}$ such that $\forall j<i: G_j=H_j$ and $G_i\ors H_i$.
By the proof of Theorem~\ref{bijection} we see that $G\ll H$ if and only if $\iota(G)<\iota(H)$.

We say that a vector $X\in \FF_q^2$  is \emph{allowable in
position $j$ for $(G_1,\ldots,G_{j-1})$} if and only if
$n_{\fq}(S_j^X,n)>0$, i.e.
 $(G_1,\dots,G_{j-1},X,X_{j+1},\dots,X_{2n+1})$ represents
  a totally $\fq$-singular line for at least one choice
 of $X_{j+1},\ldots,X_{2n+1}$.

\begin{theorem}\label{bijection}
The index function $\iota$ defined in~\eqref{iota} is a bijection.
\end{theorem}
\begin{proof}
As  $|\Delta_{n,2}|_1=|\iI|$, it is enough to show that
$\iota$ is injective.
Let
\[
\setlength{\nulldelimiterspace}{0pt}
\begin{array}{l@{=}l}
 G&(G_1, G_2, \ldots, G_{i-1}, G_i, \ldots, G_{2n+1}); \\
 H&(H_1, H_2, \ldots, H_{i-1}, H_i, \ldots, H_{2n+1}).
 \end{array} \]
We will show that if $G\ll H$, then $\iota(G)<\iota (H)$.
Suppose $G_i\ors H_i$ and $G_s=H_s\,\,\forall s<i.$
Define
\begin{multline}
\label{ee1}
 \iota^{\ors}(G):=\{ (X_1,\ldots,X_{2n+1}) : X_1\ors G_1 \}
 \cup\{ (G_1,X_2,\ldots,X_{2n+1}) : X_2\ors G_2 \} \cup \ldots\\
 \ldots \cup
\{ (G_1,G_2,\ldots,G_{2n}, X_{2n+1}) : X_{2n+1}\ors G_{2n+1} \}.
\end{multline}
  In~(\ref{ee1}) and throughout this proof, the elements of the sets are all matrices in RREF
representing totally $\fq$-singular lines. Clearly, if $G_1=H_1, \ldots, G_{i-1}=H_{i-1}$ and $G_i\ors H_i$ for
some columns $H_1,\ldots,H_i$, then
\[ G\in\{(G_1,\ldots,G_{i-1},X_i,X_{i+1},\ldots,X_{2n+1}) : X_i\ors H_i\}; \]
in particular, $G\in\iota^{\ors}(H)$.
Furthermore, if $G\in\iota^\ors(H)$, then $\iota^\ors(G)\subset\iota^\ors(H)$.
Suppose
$Y:=(Y_1,\ldots,Y_{2n+1})\in\iota^\ors(G)$. Then,
there exists $j$ such that $Y_1=G_1,\ldots,Y_{j-1}=G_{j-1}$ and $Y_j\ors G_j.$
\begin{itemize}
\item
If $j<i$, then $Y_1=G_1=H_1,\ldots,Y_{j-1}=G_{j-1}=H_{j-1}$ and
$Y_j\ors H_j=G_j$; thus $Y\in\iota^\ors(H)$.
\item
If $j=i$, then $Y_i\ors G_i\ors H_i$ and
$Y\in\{(G_1,\ldots,G_{i-1},X_i,X_{i+1},\ldots,X_{2n+1}) : X_i\ors H_i\}$;
thus $Y\in\iota^\ors (H)$.
\item
If $j>i$, then $Y_i=G_i\ors H_i$;
thus,
\[ Y\in
\{(G_1,\ldots,G_{i-1},X_i,X_{i+1},\ldots,X_{2n+1}) : X_i\ors H_i\}; \]
consequently, $Y\in\iota^\ors (H)$.
\end{itemize}
As $G\in\iota^\ors(H)$ but $G\not\in\iota^\ors (G)$, the
above inclusions are proper.
We now show that $\iota(G)=|\iota^\ors(G)|$.
Note that
\begin{multline*}
|\{(G_1,G_2,\ldots,G_{i-1},X_i,X_{i+1},\ldots X_{2n+1}) :X_i\ors G_i \}|= \\
\sum_{X_i\ors G_i}|\{(G_1,G_2,\ldots,G_{i-1},X_i,X_{i+1},\ldots X_{2n+1})\}|
=\sum_{X_i\ors G_i}n_{\fq}((G_1,\ldots,G_{i-1},X_i),n).
\end{multline*}
Furthermore, as the sets in~\eqref{ee1} are  disjoint,
\begin{multline*}
 |\iota^\ors(G)|=|\{ (X_1,\ldots,X_{2n+1}) : X_1\ors G_1 \}|+
 |\{ (G_1,X_2,\ldots,X_{2n+1}) : X_2\ors G_2 \}|+\cdots\\
 +
 |\{ (G_1,G_2,\ldots,G_{2n}, X_{2n+1}) : X_{2n+1}\ors G_{2n+1} \}|=\\
 \sum_{X_1\ors G_1}n_{\fq}((X_1),n)+\sum_{X_2\ors G_2}n_{\fq}((G_1\,X_2),n)+
 \ldots
 +\!\!\!\!\!\!\sum_{X_{2n+1}\ors G_{2n+1}} \!\!\!\!\!\!\!n_{\fq}((G_1\,G_2\,\ldots\,G_{2n}\,X_{2n+1}),n)=\\
 =\sum_{i=1}^{2n+1}\sum_{X_i\ors G_i}n_{\fq}((G_1,\ldots,G_{i-1},X_i),n)=\iota(G).
\end{multline*}
To conclude, observe that for any two distinct lines represented by matrices $G$ and $H$ in RREF we
have either $G\in\iota^\ors(H)$ or $H\in\iota^\ors(G)$. The former
yields $\iota^\ors(G)\subset\iota^\ors(H)$, whence $\iota(G)<\iota(H)$;
the latter yields, in an entirely analogous way, $\iota(G)>\iota(H)$. In any case
$G\neq H$ gives $\iota(G)\neq\iota(H)$ and $\iota$ is injective.
\end{proof}
Given an index $i\in\mathbb{I}$, the following theorem characterizes each column $G_k,\,1\leq k\leq 2n+1$, of the representative matrix $G_{\ell}$ of a totally singular line
$\ell$ as the maximum allowable vector of $\FF_{q}^2$ for the given value of $i$ and $k$. This theorem is crucial to invert the enumerative function $\iota.$
\begin{theorem}
\label{tInv}
Suppose
$G=(G_1,\dots, G_{2n+1})$ represents a totally singular line $\ell$
and let $\iota(\ell)=i$. Let also for any $k=1,\ldots, 2n+1$,
\[ \theta(G_{\leq k}):=\sum_{X\ors G_k} n_{\fq}(S_k^X,n),\qquad
 i_k:=i-\sum_{j=1}^{k-1}\theta(G_{\leq j}). \]
Then $G_k$  is the
maximum vector of $\FF_q^2$ with respect to the order $\ors$ such that
$\theta(G_{\leq k})\leq i_k$.
\end{theorem}

\begin{proof}
Define
\[ \Theta(G_{\leq k}):=\{ (G_1,\ldots,G_{k-1},X,\ldots)\in\iota^{\ors}(G) : X\ors G_k \}, \]
\[ \Lambda(G_{\leq k}):=\{ (G_1,\ldots,G_{k-1},Y,\ldots)\in\iota^\ors(G): Y\ore G_k\}. \]
Then,
\[ \Lambda(G_{\leq 1})=\{ (Y,\ldots )\in\iota^\ors(G): Y\ore G_1 \}=\iota^\ors(G). \]
We have
\[ |\Theta(G_{\leq k})|=\sum_{X\ors G_k} n_{\fq}(S_k^X,n)=\theta(G_{\leq k}). \]
On the other hand, for $k>1$ we can write
\begin{multline*}
 \Lambda(G_{\leq k})= \iota^\ors(G)\setminus\big(
  \{ (X_1,\ldots): X_1\ors G_1 \}\cup
  \{(G_1,X_2,\ldots): X_2\ors G_2\}\cup
  \cdots \\ \cdots \cup \{ (G_1,G_2,\ldots,G_{k-2},X_{k-1},\ldots) : X_{k-1}\ors G_{k-1}\}
  \big)=
  \iota^\ors (G)\setminus\bigcup_{j=1}^{k-1}\Theta(G_{\leq j}).
\end{multline*}
Thus,
\[ |\Lambda(G_{\leq k})|=\iota(G)-\sum_{j=1}^{k-1}\theta(G_{\leq j})=i_k. \]
We distinguish two cases:
\begin{itemize}
\item $k=1$. By way of contradiction, suppose $G_1$ is not maximum and
   $\theta(G_{\leq1})\leq i_1=i$.
Then, there is an element $G'_1\in \FF_q^2$,  with $G_1\ors G_1'$ and
$\theta(G_{\leq 1}')\leq i$. By construction,
$\Lambda(G_{\leq 1})\subset\Theta(G_{\leq1}')$. Observe that $G\in\Theta(G_{\leq1}')$ but
$G\not\in\Lambda(G_{\leq1})$. Thus, the inclusion is proper.
Moving to the cardinalities we have
\[ i=|\Lambda(G_{\leq1})|<|\Theta(G_{\leq1}')|=\theta(G_{\leq1}')\leq i, \]
a contradiction.
\item $k>1$.  Suppose that the thesis holds for $j\leq k$ but not for $j=k+1$, i.e.\
  all  $G_j$ for $j\leq k$
  are maximum such that $\theta(G_{\leq j})\leq i_j$ and
  $G_{k+1}$ is not the maximum element such that $\theta(G_{\leq k+1})\leq i_{k+1}$.
  Then, as before, there is a
  $G_{k+1}'$ such that $G_{k+1}\ors G_{k+1}'$ with $\theta(G_{\leq k+1}')\leq i_{k+1}$.
  For any $Y\ore G_{k+1}$ we have $Y\ors G_{k+1}'$; thus, the following
  holds
  \begin{multline*}
    \Lambda(G_{\leq k+1})=\{ (G_1,\ldots,G_{k},Y,\ldots)\in\iota^\ors(G):
  Y\ore G_{k+1}\}\subset \\
  \subset\{ (G_1,\ldots,G_{k},X,\ldots)\in\iota^{\ors}(G): X\ors G_{k+1}' \}=\Theta(G_{\leq k+1}').
  \end{multline*}
  Furthermore, as $G\in\Theta(G_{\leq k+1}')$ but $G\not\in \Lambda(G_{\leq k+1})$,
  the above inclusion is proper.
  Thus,
  \[ i_{k+1}=|\Lambda(G_{\leq k+1})|<|\Theta(G_{\leq k+1}')|=\theta(G_{\leq k+1}')\leq i_{k+1},
  \]
  a contradiction.
\end{itemize}
\end{proof}
In Table~\ref{tt} we show in detail the procedure arising from
Theorem~\ref{tInv} to efficiently invert the function
$\iota$.
Observe that the check
$n_{\fq}(S_k^Y,n)>0$ is necessary, as each column $G_k$ must
be allowable and
columns which are allowable in a given position $k$ may not be allowable in
position $k-1$ or vice-versa.
\begin{table}[h]
\caption{Inverse of $\iota$}
\begin{algorithmic}
\Require $i\in\{0,\dots, N-1\}$
\State $i_1\gets 1$
\For {$k=1,\ldots,2n+1$}
\State
 $M\gets\{Y\colon \sum_{X\ors Y} n_{\fq}(S_k^X,n)\leq i_k
 \text{ and } n_{\fq}(S_k^Y,n)>0 \}$
\State
 $G_k\gets \max M$
\State $\theta(G_{\leq k}) \gets \sum_{X\ors G_k} n_{\fq}(S_k^X,n)$;
\State $i_{k+1}\gets i_{k}-\theta(G_{\leq k})$;
\EndFor
\State \Return $G=(G_1,\ldots,G_k,\ldots,G_{2n+1})$
\end{algorithmic}
\label{tt}
\end{table}

\subsection{Complexity}
\label{compl}
    We now estimate the actual cost of the enumerative encoding presented in this section.
    In the orthogonal case, to evaluate $\iota$ we need to compute at most $q^2-1$ values
    of $n_{\fq}(S_j^X,n)$ for any $j=1,\ldots,2n+1$ as $X$ varies in $\FF_q^2$.
    So, the overall complexity turns out to be $O(q^2n^3)$.
    Conversely, given an index $i\in\iI$, recovering the corresponding line $\ell=\iota^{-1}(i)$
    requires to  test at most $q^2-1$ vectors $X\in\FF_q^2$ for each column of $G$; thus, the
    cost is once more $O(q^2n^3)$.
    In the symplectic case,
    the same arguments  give a complexity of $O(q^2n^2)$ for enumerative enconding.

   The computational complexity of the enumerative algorithm for the orthogonal and symplectic Grassmannians are
   summarized in Theorem~\ref{main thm}.

\section{Application to polar Grassmann codes }
\label{s5}
We now apply the enumeration techniques discussed in the previous
sections to efficiently implement the polar Grassmann  codes $\cP_{n,2}$
and $\cW_{n,2}$.
 We refer to Section~\ref{PGC} for the definition and some basics about
 these codes.

We shall focus the discussion on  the case of orthogonal
polar Grassmann codes, while we will only point out the adjustements
to be made for the symplectic case $\cW_{n,2}$, as the arguments in the
two cases are very similar.
\subsection{Encoding}
\label{s-enc}
As in Section~\ref{s2}, let $V$ be a vector space of
dimension $2n+1$ over $\FF_q$ and fix a basis
$B:=(e_1,\ldots,e_{2n+1})$ of $V.$

It is well known that the dual $(\bigwedge^kV)^*$
of the vector space $\bigwedge^kV$ is isomorphic to $\bigwedge^{2n+1-k}V$.
We recall the following universal property of the $k^{\text{th}}$-exterior
power of a vector space.
\begin{theorem}[\,{\cite[Theorem 14.23]{ALA}}]
\label{ut}
Let $V,U$ be two vector spaces over the same field.
 A map $f: V^k \longrightarrow U$ is alternating
  $k$--linear if and only if there is a
  linear map $\zeta:
  \bigwedge^k V \longrightarrow U$ with $\zeta(v_1 \wedge v_2
  \wedge \cdots \wedge v_k) = f(v_1,v_2,\ldots,v_k)$.
  The map $\zeta$ is uniquely determined.
\end{theorem}
By Theorem~\ref{ut}, for $k=2,$
any linear functional $\zeta$ on $\bigwedge^2 V$ corresponds
to a bilinear alternating form on $V$; hence it can be represented by an antisymmetric
$(2n+1)\times (2n+1)$-matrix $M=(\fmm_{ij})_{i,j=1}^{2n+1}$ whose entries
are
$\fmm_{ij}=\zeta(e_i\wedge e_j)$.
With a slight abuse of notation, for any $\ell\in\Delta_{2,n}$
write $\zeta(\ell):=\zeta(G_{1}^{\ell}\wedge G_2^{\ell})$,
where $G_{1}^{\ell}$ and $G_2^{\ell}$ are the two rows of the representative matrix $G_{\ell}$ of $\ell$ in RREF.
Clearly, we also have $\zeta(\ell)=G_1^{\ell}MG_2^{\ell\,T}$.
Let $\cP_{n,2}$ be the line orthogonal Grassman $[N,K,d]$-code as in Theorem~\ref{pgc-thm}.

\begin{definition}\label{encoding function}
Let $\psi\colon \FF_q^K\rightarrow \cP_{n,2}$ be the function mapping any message $(m_i)_{i=1}^K\in \FF_q^K$ to the codeword
$(c_{i+1})_{i=0}^{N-1}\in \cP_{n,2}$ where   \[ c_{i+1}:=\zeta(\iota^{-1}(i))={G^{(i)}_1}M{G^{(i)T}_2}, \]

\noindent the function $\iota$ is defined in~(\ref{iota}), $G^{(i)}_1$ and $G^{(i)}_2$ are the rows of the representative matrix $G_{\ell_i}$ in RREF
  of the line $\ell_i:=\iota^{-1}(i)$
and $\zeta$ is the bilinear alternating form with matrix $M:=M_0-M_0^T$ with respect to $B$,
  with
  \begin{itemize}
    \item for $q$ odd:
\[ M_0=\begin{pmatrix}
  0 & m_1 & m_2  & \ldots & m_{2n} \\
  0 &  0  & m_{2n+1} & \ldots & m_{4n-1} \\
  \vdots & & \ddots & \ddots & \vdots \\
  0 & 0   & \ldots & 0 & m_{n(2n+1)} \\
  0 & \ldots & \ldots & 0 & 0
\end{pmatrix} \]
and
\item for $q$ even:
\[ M_0=\begin{pmatrix}
  0 & m_1 & m_2 &   & \ldots & m_{2n} \\
  0 &  0  & m_{2n+1}  & & \ldots & m_{4n-1} \\
  \vdots & & \ddots & & \ddots & \vdots \\
  0 & 0 & \ldots & 0 &m_{n(2n+1)-2} & m_{n(2n+1)-1}\\
  0 & 0   &  \ldots & 0 & 0 & 0 \\
  0 & \ldots & \ldots & 0 & 0 & 0
\end{pmatrix}. \]
\end{itemize}
\end{definition}
Clearly, $\psi$ is a linear function.

\begin{theorem}\label{encoding thm}
The function $\psi$ introduced in Definition~\ref{encoding function} is an encoding for $\cP_{n,2}.$
 \end{theorem}
\begin{proof}
We need to prove that $\psi$ is injective.
We first point out that, by Definition~\ref{encoding function}, a codeword $\mathbf{c}\in\cP_{n,2}$ is associated with a message $\mathbf{m}$ if the positions $c_{i+1},\,\, 0\leq i\leq N-1,$
of $\mathbf{c}$ are the values assumed on the lines of $\Delta_{n,2}$ by the linear functional $\zeta\in (\bigwedge^2 V)^*$ defined by $\mathbf{m}.$
In order to make more explicit the link between codewords $\mathbf{c}$ and linear functionals $\zeta$, we shall write $\mathbf{c}_{\zeta}.$

Suppose that $q$ is odd. Then $\cP_{n,2}=\bigwedge^2V$ (see \cite{IP13}).
A codeword $\mathbf{c}_{\zeta}$ is null if and only if the functional
$\zeta$ is identically null on $\bigwedge^2V$.
Indeed, since  $\langle\varepsilon_2(\Delta_{n,2})\rangle=\bigwedge^2V$,
 there exist
some positions $i_1+1,\ldots,i_{n(2n+1)}+1$ of $\mathbf{c}_{\zeta}$ such that $B':=(\iota^{-1}({i_j}))_{j=1}^{n(2n+1)}$
is a basis of $\bigwedge^2V.$ If $\mathbf{c}_{\zeta}=\mathbf{0}$, then  $\zeta$ is zero on all the elements
of $B'$. Hence, by linearity, $\zeta$ is the null functional.
This proves that $\psi$ is injective.

Suppose that $q$ is even. Now $\langle\varepsilon_2(\Delta_{n,2})\rangle$ is a hyperplane $W$ of $\bigwedge^2 V$ (see \cite{IP13}).
By Definition~\ref{encoding function}, any non-null codeword $\mathbf{c}_{\zeta}$ is associated with a message
$\mathbf{m}\in \FF_q^K$ which is in correspondence with a linear functional $\zeta\in (\bigwedge^2 V)^*$ which is not identically null on $W.$
By~\cite{IP14}, $W=\ker\zeta_0$ with
\begin{equation*}
\zeta_0\colon \bigwedge^2 V\rightarrow \FF_q,\,\,\,\, \zeta_0((u_{ij})_{1\leq i<j\leq 2n+1}):=u_{23}+u_{45}+\cdots+u_{2n,2n+1}
\end{equation*}
 where $u_{ij}$ are the Pl\"ucker coordinates of vectors in  $\bigwedge^2 V$ with respect to the basis $(e_i\wedge e_j)_{1\leq i<j\leq 2n+1}.$

 Given any message $\mathbf{m}=(m_1,\ldots,m_{n(2n+1)-1})\in\FF_q^K$, by Definition~\ref{encoding function},
 the antisymmetric matrix $M$ defined by $\mathbf{m}$ has $\fmm_{2n,2n+1}=0$.

 Let $\zeta\in(\bigwedge^2V)^*$ be the functional associated with $\mathbf{m}$ by means of
the bilinear alternating form defined by $M$.
 Observe that the functional $\zeta|_{W}$ induced by
the restriction of $\zeta$ to $W$ is null if and only if $\zeta$ is proportional to $\zeta_0$.
By construction, we have
 $\zeta(e_{2n}\wedge e_{2n+1})=0$.
We are now ready to prove that the function $\psi:\FF_q^K\to\cP_{n,2}$ is injective.
Indeed, if $\mathbf{m}_1$ and $\mathbf{m}_2$ are  messages inducing
two forms $\zeta_1$ and $\zeta_2$ and such
that $\psi(\mathbf{m}_1)=\psi(\mathbf{m}_2)$, then
$\zeta'=\zeta_1-\zeta_2$ is null on $W$. So, $\zeta'$
must be proportional to $\zeta_0$.
However, $(\zeta_1-\zeta_2)(e_{2n}\wedge e_{2n+1})=0$,
while $\zeta_0(e_{2n}\wedge e_{2n+1})=1$. So, the only possibility is that
$\zeta_1-\zeta_2$ is the null functional on $\bigwedge^2V$, i.e. $\mathbf{m}_1=\mathbf{m}_2$.
\end{proof}
A straightforward counting argument provides the following relationship between
the components of $\mathbf{m}$ and the entries $\fmm_{ij}$ of $M$ with $1\leq i<j\leq 2n+1$:
\begin{equation}
\label{components}
 {\fmm}_{ij}=m_{2n(i-1)+j-\frac{i^2-i}{2}-1}.
\end{equation}
So, we can directly determine
each component of $\mathbf{c}$ using just $\mathbf{m}$ without having to
resort to the generator matrix of the code.

\medskip

The case of symplectic Grassmann codes $\cW_{n,2}$ is entirely analogous to what
we proposed for $\cP_{n,2}$ for $q$ even. Definition~\ref{encoding function} remains the same if we write $\cW_{n,2}$ in place of $\cP_{n,2},$ ${{2n}\choose{k}}-1$ in place of ${{2n+1}\choose{k}}-1$ for $K$ and we
consider $q$ arbitrary. The only further difference in Theorem~\ref{encoding thm} is that $\langle\varepsilon_k(\overline{\Delta}_{n,2})\rangle=\ker(\zeta_0)$ with
\begin{equation*}
\zeta_0\colon \bigwedge^2 V\rightarrow \FF_q,\,\,\,\, \zeta_0((u_{ij})_{1\leq i<j\leq 2n+1}):=u_{12}+u_{34}+\cdots+u_{{2n-1},2n},
\end{equation*}
 where $u_{ij}$ are the Pl\"ucker coordinates of vectors in  $\bigwedge^2 V$.

\subsection{Decoding}
\label{s-dec}
 We address now the problem of recovering the original
message $\mathbf{m}$ once a codeword $\mathbf{c}$ has been given.
Here we shall assume that no error occurred; how to perform
error correction shall be discussed in the next section.
Clearly, by Section~\ref{s-enc}, recovering $\mathbf{m}$ is equivalent to reconstructing the
antisymmetric matrix $M$ associated with $\mathbf{m}$.

In general, polar Grassmann codes are not systematic, nor
there are entries in $\mathbf{c}$ corresponding exactly to the values $\fmm_{ij}$ in $M$.
None the less, it is possible to provide a list of lines yielding information positions in
 $\mathbf{c}$ such that the values $\fmm_{ij}$ can be easily determined.
\begin{theorem}
\label{Decode}
Let $\mathbf{c}$ be a codeword of $\cP_{n,2}$ and $M=({\fmm}_{ij})_{1\leq i,j\leq 2n+1}$ be the antisymmetric matrix associated
with the message $\mathbf{m}$ mapped to $\mathbf{c}$ using the
encoding $\psi$. Suppose that the pair $(i,j)$ with $1\leq i<j\leq 2n+1$ is in one of the following types:
\begin{enumerate}[Type I:]
\item $i\geq 2$ even and $j\geq i+2$
  or $i$ odd and $j\geq i+1$;
\item $i\geq 2$ even and $j=i+1$;
\item $i=1$ and $j>i$.
\end{enumerate}
Then the following holds:
\begin{itemize}
\item If $(i,j)$ is of Type I then ${\fmm}_{ij}=c_{\iota(\ell_{i,j})+1}$ where
$\ell_{i,j}:=\langle e_i,e_j\rangle$.
\item If $(i,j)$ is of Type II then ${\fmm}_{ij}$  can be obtained by solving a system of $2$ linear equations in $2$ unknowns for $q$ odd
  and a single linear equation for $q$ even.
\item If $(i,j)$ is of Type III then ${\fmm}_{ij}$  can be obtained by solving a linear equation.
\end{itemize}
\end{theorem}
\begin{proof}
If $(i,j)$ is of Type I, the thesis is straightforward. If $(i,j)$ is of Type II, we distinguish two cases according to whether
$q$ is odd or even.

Suppose $q$ odd. Consider two lines
$\ell^1:=\langle e_i+e_{i+3}, e_{i+1}-e_{i+2}\rangle$ and
$\ell^2:=\langle e_i-e_{i+3},e_{i+1}+e_{i+2}\rangle$ and
call $c_x:=c_{\iota(\ell^1)+1}$, $c_y:=c_{\iota(\ell^2)+1}$ the corresponding entries of $\mathbf{c}$.
Then we have
\[ \begin{cases}
  {\fmm}_{i,i+1}-{\fmm}_{i,i+2}-{\fmm}_{i+1,i+3}+{\fmm}_{i+2,i+3}=c_x \\
  {\fmm}_{i,i+1}+{\fmm}_{i,i+2}+{\fmm}_{i+1,i+3}-\mathfrak{m}_{i+2,i+3}=c_y
   \end{cases}. \]
The entries ${\fmm}_{i+1,i+3}$ and ${\fmm}_{i,i+2}$ correspond to indexes of Type I;
thus they  can be read off $\mathbf{c}$ directly.
The remaining unknowns
$\mathfrak{m}_{i,i+1}$ and $\mathfrak{m}_{i+2,i+3}$ can now be recovered by solving a system
in two unknowns. Observe that this operation has fixed complexity $O(1)$.

Suppose $q$ even. Recall that, in this case, $\fmm_{2n,2n+1}=0$.
Consider the line $\ell:=\langle e_i+e_{2n},e_{i+1}+e_{2n+1}\rangle$ and
let $c_x=c_{\iota(\ell)+1}$.
We get
\[ \fmm_{i,i+1}+\fmm_{i,2n+1}+\fmm_{i+1,2n}=c_x. \]
As $\fmm_{i,2n+1}$ and $\fmm_{i+1,2n}$ correspond to indexes of Type I, this gives the value of $\fmm_{i,i+1}$.

Suppose $(i,j)=(1,j)$ is of Type III.
If $j>3$, we consider the line $\ell=\langle e_1-e_2+e_3,e_j\rangle$.
A straightforward computation shows that the corresponding entry $c_z:=c_{\iota(\ell)+1}$ is
\[ \mathfrak{m}_{1j}-\mathfrak{m}_{2j}+\mathfrak{m}_{3j}=c_z \]
and both $(2,j)$ and $(3,j)$ are of Type I; thus we just have to solve this
equation.
As for the remaining
coefficients ${\fmm}_{12}$ and ${\fmm}_{13}$, we use the entries
corresponding to $\ell^{12}=\langle e_1-e_4+e_5,e_2\rangle$ and
$\ell^{13}=\langle e_1-e_4+e_5,e_3\rangle.$
\end{proof}

Theorem~\ref{Decode} shows that it is possible to extract any component of the message
$\mathbf{m}$ from a codeword $\mathbf{c}$ with complexity $O(1)$. As such, the complexity
to recover the whole of $\mathbf{m}$ is $O(n^2)$.

In the symplectic case, the same arguments as in Theorem~\ref{Decode}
lead to the following.
\begin{theorem}
Let $\mathbf{c}$ be a codeword of $\cW_{n,2}$ and $M=({\fmm}_{ij})_{1\leq i,j\leq 2n+1}$ be the antisymmetric matrix associated
with the message $\mathbf{m}$  mapped to  $\mathbf{c}$ using the
encoding $\psi$.
 Suppose that the pair $(i,j)$ with $1\leq i<j\leq 2n+1$ is in one of the following types:
\begin{enumerate}[Type I:]
\item $i\geq 1$ odd and $j\geq i+2$
  or $i$ even and $j\geq i+1$;
\item $i\geq 1$ odd and $j=i+1$;
\end{enumerate}
Then the following holds:
\begin{itemize}
\item If $(i,j)$ is of Type I then ${\fmm}_{ij}=c_{\iota(\ell_{i,j})+1}$ where
$\ell_{i,j}:=\langle e_i,e_j\rangle$.
\item If $(i,j)$ is of Type II then ${\fmm}_{ij}$  can be obtained by solving one linear equation.
\end{itemize}

\end{theorem}

\subsection{Error correction}
Locally decodable codes have received much attention in
recent years; see~\cite{LDC1,LDC2} for some surveys.
In general, a code is \emph{locally decodable} if it
is able to recover a given component $m_i$
of a message $\mathbf{m}$ with probability larger than $1/2$ querying just
a fixed number of components $r_j$ of the received vector $\mathbf{r}$ ---
this, clearly, under the assumption that
not too many errors have occurred; see~\cite{KT}.

In this section we shall introduce an algorithm to
reconstruct a correct information position using only
some local information. 

Let $\mathbf{r}=(r_{i+1})_{i=0}^{N-1}\in\FF_q^N$ be a received vector; by the
arguments in Section~\ref{s-enc}, $\mathbf{r}$ is
a codeword $\mathbf{c}$ if and only if there exists $\zeta\in(\bigwedge^2V)^*$,
such that $r_{i+1}=\zeta(G_1^{\ell}\wedge G_2^{\ell})$ for any $0\leq i\leq N-1$ where $\ell=\langle G_{1}^{\ell},G_2^{\ell}\rangle$
and $\ell=\iota^{-1}(i)$.
As in \S~\ref{s-enc}, we shall write $r_{i+1}:=\zeta(\ell).$

Fix now a position $i+1$ and consider the totally singular line $\ell=\iota^{-1}(i).$ Let
\[ \Sigma_{\ell}:=\{ \pi : \ell\subseteq\pi\subseteq \mathcal{Q}, \dim\pi=3 \} \]
be the set of all totally singular planes of $\mathcal{Q}$ containing $\ell$ (we remind
that we have always used vector dimension throughout the paper, but
we adopt projective terminology when speaking of geometric objects).
The restriction $\zeta_{\pi}$ of $\zeta$ to each plane $\pi$ determines by Theorem~\ref{ut} a
degenerate alternating bilinear form which we shall still denote by the same symbol.
Clearly, in absence of errors, all the forms $\zeta_{\pi}$, as $\pi$ varies in $\Sigma_{\ell}$, must agree on
$\ell$.
The overall number of totally singular planes containing a given fixed line
$\ell$ is $|\Sigma_{\ell}|=\frac{q^{2n-4}-1}{q-1}$; this is also the total number of different
forms $\zeta_{\pi}$ we can consider.

Let $0<\varepsilon\leq |\Sigma_{\ell}|$ be a parameter denoting the number of planes of
$\Sigma_{\ell}$ we want to use and
consider the following error correction strategy:
\begin{enumerate}
\item\label{beginA}
 Choose $\varepsilon\leq |\Sigma_{\ell}|$ planes at random in $\Sigma_{\ell}$.
\item
 For each of the chosen planes, say $\pi$, let $p,q,s$ be three distinct lines of $\pi$ different
 from $\ell$ forming a triangle and recover the alternating bilinear form $\phi^{\pi}$
 such that $\phi^{\pi}(p)=r_{\iota(p)+1}$, $\phi^{\pi}(q)=r_{\iota(q)+1}$, $\phi^{\pi}(s)=r_{\iota(s)+1}$.
 This corresponds
  to solving a linear system of $3$ equations in $3$ unknowns.
\item
  If all forms $\phi^{\pi}$ are such that $\phi^{\pi}(\ell)=r_{\iota(\ell)+1}$, then we claim that
  the value $r_{\iota(\ell)+1}$ is correct and set $c_{\iota(\ell)+1}=r_{\iota(\ell)+1}$;
  otherwise, take $c_{\iota(\ell)+1}$ as the value assumed by the majority of the forms $\phi^{\pi}$, as $\pi$ varies in $\Sigma_{\ell}$, when evaluated on $\ell$.
\end{enumerate}

The same correction strategy can be implemented for polar symplectic Grassmann
codes, by considering totally isotropic planes instead of totally singular
ones.

\section*{Acknowledgments}
This research  was performed within the activity of GNSAGA of INdAM (Italy) whose support both authors acknowledge.

\end{document}